\documentclass[10pt]{amsart}

\usepackage[a4paper,margin=1.1in]{geometry}

\usepackage{amsmath,amssymb,amsthm,mathtools}
\usepackage[T1]{fontenc}
\usepackage{microtype}

\usepackage{graphicx}
\usepackage{booktabs}
\usepackage{xcolor}
\usepackage{titlesec}
\usepackage{footnote}
\usepackage{algorithm}
\usepackage{algpseudocode}
\usepackage{mathrsfs}
\usepackage{adjustbox}
\usepackage{threeparttable}

\titleformat{\section}
  {\normalfont\Large\bfseries}
  {\thesection}
  {1em}
  {}

\titleformat{\subsection}
  {\normalfont\large\bfseries}
  {\thesubsection}
  {1em}
  {}

\titleformat{\subsubsection}
  {\normalfont\normalsize\bfseries}
  {\thesubsubsection}
  {1em}
  {}

\usepackage[
    colorlinks=true,
    linkcolor=blue,
    citecolor=blue,
    urlcolor=blue
]{hyperref}

\newtheorem{theorem}{Theorem}[section]
\newtheorem{proposition}[theorem]{Proposition}
\newtheorem{lemma}[theorem]{Lemma}

\theoremstyle{definition}
\newtheorem{definition}[theorem]{Definition}
\newtheorem{assumption}[theorem]{Assumption}

\theoremstyle{remark}
\newtheorem{remark}[theorem]{Remark}

\makeatletter
\renewcommand\paragraph{\@startsection{paragraph}{4}{0pt}%
                                    {0.85ex \@plus0ex \@minus.2ex}%
                                    {-1em}%
                                    {\normalfont\normalsize\bfseries}}
\renewcommand\subparagraph{\@startsection{subparagraph}{5}{\parindent}%
                                       {3.25ex \@plus1ex \@minus .2ex}%
                                       {-1em}%
                                      {\normalfont\normalsize\bfseries}}
\makeatother

\newcommand{\ccA}{{\mathscr A}}  
\newcommand{\ccL}{{\mathscr L}}
\newcommand{\ccB}{{\mathscr B}}

\newcommand{\ccF}{{\mathscr F}}

\newcommand{\ccH}{{\mathscr H}}
\newcommand{\ccP}{{\mathscr P}}
\newcommand{\ccE}{{\mathscr E}}

\newcommand{\ccN}{{\mathscr N}}

\renewcommand{\ccP}{{\mathscr P}}

\usepackage{titlesec}

\title[]{Filtering Credit Risk with Stochastic Discontinuities}

\author{Felix B. Tambe-Ndonfack$^*$}\thanks{$^*$Department of Mathematical Stochastics, Mathematical Institute, University of Freiburg, Ernst-Zermelo-Str. 1, 79104, Freiburg, Germany. Email: \texttt{felix.ndonfack@stochastik.uni-freiburg.de}}

\begin{document}
\maketitle
\begin{abstract}
    We develop a structural credit-risk model under incomplete information in which investors observe firm value only indirectly through noisy market signals and scheduled corporate disclosures. While disclosure dates are known in advance, their informational content is random, leading to stochastic discontinuities in the observation process. We derive the Kushner-Stratonovich equation for structural credit-risk models with endogenous default by applying the nonlinear filtering framework with predictable jumps. We then study the valuation and local risk-minimization hedging for default-sensitive securities under partial information. The interaction between predictable disclosure events and endogenous default produces discrete adjustments in the conditional default compensator, leading to announcement-driven distortions in credits spreads and hedge ratios that are absent from classical diffusion-based and inaccessible-jump models. Numerical experiments illustrate how scheduled disclosures affect filtered default probabilities, Credit Default Swaps (CDS) spreads, and hedging strategies, generating characteristic pre-announcement dynamics in credit spreads. 
\end{abstract}
\medskip

\noindent\textbf{Keywords:}
Credit risk; Incomplete information; Stochastic discontinuities;
Kushner--Stratonovich equation; Filtering; Hedging.

\medskip

\noindent\textbf{2020 Mathematics Subject Classification:}
91G40, 93C41, 60G35, 91G20.

\section{Introduction}
Corporate credit markets are characterized by two fundamental sources of uncertainty. First, investors do not observe the firm's asset value directly and must infer its financial condition from imperfect public information. Second, economically significant information is released through scheduled events such as earnings announcements, credit-rating reviews, regulatory filings, and covenant tests. Although the timing of these disclosures is known in advance, their informational content is uncertain and frequently leads to abrupt revisions in market assessments of firm value, default risk, and credit spreads. Incorporating these predictable information releases into structural credit-risk models is therefore essential for understanding the dynamics of defaultable securities under incomplete information.

Existing structural credit-risk models capture only part of this mechanism. Classical incomplete-information models describe how investors learn about an unobservable firm value from noisy observations but assume continuous asset dynamics, thereby excluding abrupt revisions generated by scheduled disclosures. Conversely, structural jump-diffusion models allow discontinuous changes in firm value but typically model jump times as totally inaccessible stopping times, making them unsuitable for representing publicly announced information releases. Consequently, the interaction between endogenous default, predictable disclosure events, and partial observation remains largely unexplored.

In this paper, we develop a structural credit-risk model under incomplete information in which investors observe the firm only indirectly through noisy accounting signals and scheduled public disclosures. Disclosure dates are deterministic, or more generally predictable, while the informational content of each announcement is random. As a consequence, both the latent firm value and the observable information process exhibit stochastic discontinuities at predictable times. This naturally leads to a nonlinear filtering problem with predictable jumps. Building on the general filtering framework of Schmidt and Tambe-Ndonfack \cite{schmidt2026nonlinear}, we derive the Kushner-Stratonovich equation corresponding the conditional distribution of the firm's asset value and apply it to the valuation and hedging of default-sensitive securities.

\subsection{Related literature}
The literature most closely related to this paper can be organized around three strands: structural credit-risk models, incomplete information and nonlinear filtering, discontinuous asset dynamics and predictable stochastic discontinuities, and hedging under restricted information. Each strand provides an essential component of our framework.

The first concerns structural credit-risk models, in which default is modeled endogenously through the dynamics of the firm's assets. Beginning with the seminal contributions of Merton \cite{merton1974pricing} and Black and Cox \cite{black1976valuing}, these models define default endogenously through firm-value dynamics and have become a central framework for the analysis of corporate credit risk. A comprehensive treatment of structural credit-risk modeling, including valuation and hedging of credit risk-sensitive securities, is provided by Bielecki and Rutkowski \cite{bielecki2002credit}. While structurally appealing, classical complete-information models have well-known empirical limitations; for example  Collin-Dufresne et al. \cite{collin2001determinants} show that leverage and asset volatility explain only a limited fraction of observed credit-spread variation.

The second strand introduces incomplete information into structural credit-risk models. Duffie and Lando \cite{duffie2001term} show that noisy accounting signals generate strictly positive short-term credit spreads. Guo et al. \cite{guo2009credit} develop credit-risk models with incomplete information and demonstrate how filtering techniques can be used to characterize default probabilities and credit spreads when relevant state variables are only partially observed. Frey and Schmidt \cite{frey2009pricing} develop a continuous-time structural model under partial information using nonlinear filtering techniques, which is further extended in Frey et al. \cite{frey2019corporate}. Frey and Runggaldier \cite{frey2010pricing} similarly develop a nonlinear-filtering approach to pricing credit derivatives under incomplete information, emphasizing the role of the investor's conditional distribution of the latent firm value in valuation. The broader role of filtering and incomplete information in credit-risk modeling is discussed by Frey and Schmidt \cite{frey2011filtering}, while Fontana and Runggaldier \cite{fontana2010credit} study filtering and parameter estimation in credit-risk models with incomplete information.
More recently, Schmidt and Novikov \cite{schmidt2008structural} consider an unobserved default boundary as an additional source of incomplete information. These contributions establish nonlinear filtering as a natural framework for structural credit-risk models under partial observation. However, the information structure in these models is generally driven by a continuous observation or standard default-related signals and does not explicitly capture predictable disclosure events that generate state-dependent discontinuities and discrete Bayesian updates in investors' beliefs. Our work extends this literature by incorporating scheduled information releases directly into the dynamics of the latent firm value and the observation process.  

The third strand considers structural models with jump dynamics. Jump-diffusion specifications such as Carr and Linetsky \cite{carr2006jump} and Linetsky \cite{linetsky2006pricing} allow discontinuous asset-value dynamics but typically assume that jumps occur at totally inaccessible stopping times generated by Poisson random measures. Consequently, these models cannot represent situations in which the timing of information releases is publicly known in advance. Closely related is the literature on stochastic discontinuities in term-structure modeling. In particular Gehmlich and Schmidt \cite{gehmlich2018dynamic} introduce predictable discontinuities through the default boundary, while Fontana et al. \cite{fontana2024term} and Fontana and Schmidt \cite{fontana2018general} develop multiple-curve and general term structure models with stochastic discontinuities. Our approach is complementary: we introduce predictable informational shocks directly into the firm's asset-value dynamics while preserving the first-passage structural framework of Frey and Schmidt \cite{frey2009pricing}.

The fourth strand concerns hedging under incomplete or restricted information. When investors do not have access to the full market filtration, the resulting market is generally incomplete from their perspective, and classical replication arguments may no longer apply. This motivates hedging approaches based on local risk minimization and orthogonal decompositions under the investor's restricted information. Ceci et al. \cite{zbMATH06514479} study local risk minimization under restricted information on asset prices and provide a general framework for constructing hedging strategies when the investor observes only partial market information. This perspective is closely related to the present setting, where the latent firm value is unobservable and investors must base their hedging decisions on the filtered conditional distribution generated by noisy observations and scheduled disclosures.

The four strands above provide, respectively, the endogenous default mechanism, the nonlinear-filtering framework for partial information, the modeling of discontinuities in financial dynamics, and the appropriate hedging methodology under restricted information. However they are generally developed separately. In particular, existing structural credit-risk models under incomplete information typically do not incorporate predictable disclosure-driven discontinuities, while models with jumps generally do not address the filtering problem induced by publicly known jump times and partially observed jump sizes. As a result, the joint interaction between endogenous first-passage default, incomplete information, predictable jump times, and hedging under the investor's restricted filtration remains largely unexplored.

The present paper brings these strands together by specializing the nonlinear filtering framework with predictable jump times developed by Schmidt and Tambe-Ndonfack \cite{schmidt2026nonlinear} to a structural credit-risk setting with endogenous first-passage default. This provides a unified framework for pricing and hedging under incomplete information in the presence of predictable disclosure events. In particular, the model combines three features that are typically treated separately in the existing literature: endogenous default through a first-passage mechanism, incomplete information about the firm's latent asset value, and stochastic discontinuities occurring at predictable disclosure dates.

\subsection{Main Contribution}
The paper makes four principal contributions. First, we develop a structural credit-risk model that incorporate predictable disclosures events into the dynamics of the firm's asset value while preserving the endogenous first-passage default mechanism. Unlike existing structural filtering models, the proposed framework captures discrete revisions in firm value and investor beliefs generated by scheduled information releases. Second, we derive the nonlinear filtering equations governing investors' conditional beliefs. By specializing the general Kushner-Stratonovich equation with predictable jumps to the structural credit-risk setting, we obtain a recursive characterization of the conditional distribution of the latent firm value. This conditional distribution becomes the fundamental state variable for pricing and risk management under incomplete information. Third, we characterize the valuation of default-sensitive securities under partial information. Equity values, survival probabilities, defaultable bond prices, and CDS spreads are expressed as conditional expectations with respect to the filtered distribution, thereby extending classical structural pricing formulas to markets with predictable informational shocks. Finally, we study hedging in this incomplete information environment. Since the firm's asset value is unobservable and disclosure shocks are only partially predictable, perfect replication is generally impossible. We therefore develop locally risk-minimizing hedging strategies based on the Galtchouk-Kunita-Watanabe decomposition under the investor filtration and propose a particle-filter algorithm for their numerical implementation.

The model generates several economically relevant implications. Scheduled disclosures events induce discrete Bayesian updates in investors' beliefs, producing announcement-driven adjustments in default probabilities, credit spreads and hedge ratios that are absent from both continuous-diffusion models and structural models with totally inaccessible jumps. The numerical experiments illustrate these mechanisms and quantify their effects on the pricing and risk management of corporate securities. 

\subsection{Organization of the Paper}
The remainder of the paper is organized as follows: Section
\ref{sec:Model} introduces the structural credit-risk model and the information structure; Section \ref{sec: full information} studies the valuation of default sensitive securities under complete information; Section \ref{sec : the incomplete market} develops the nonlinear filtering equations and derives pricing formulas under partial information; Section \ref{sec 5: Hedging} studies local risk-minimization hedging under incomplete information. Section \ref{sec:numerical} presents the particle filter algorithm and numerical illustrations. Technical proofs are collected in the Appendix.

\section{The Model}\label{sec:Model}

We work on a filtered probability space $(\Omega,\ccF,\mathbb F=(\ccF_t)_{t\ge 0},\mathbb P)$, satisfying the usual conditions, where $\mathbb P$ denotes the physical (probability) measure and and $\mathbb F$ represents the full-information filtration generated by the firm's underlying sate variables. We assume that the financial market is arbitrage-free, so that there exists an equivalent measure $\mathbb{Q}\sim \mathbb{P}$. Under $\mathbb Q$, all discounted traded asset prices are local martingales. Throughout the pricing and hedging analysis, expectations are taken under the risk-neutral measure $\mathbb Q$, whereas $\mathbb P$ is used to describe the real-world dynamics whenever appropriate.  We also assume a constant risk-free interest rate $r\ge 0$. The model is specified by four primitives: the firm asset value $V=(V_t)_{t \ge 0}$, the liability (default) barrier $K=(K_t)_{t\ge 0}$, the cumulative dividend process $D=(D_t)_{t \ge 0}$ and an increasing sequence of predictable announcement dates $0<T_1<T_2<\ldots$. Following \cite{frey2009pricing}, default occur when the asset value  first falls below the liability barrier,
$$\tau = \inf\{t>0: V_t \le K_t\},$$
where both $V$ and $K$ are assumed càdlàg with strictly positive initial conditions. Under the assumptions stated below, $\tau$ is an $\mathbb F$-stopping time.
For simplicity, dividend payments and public disclosures occur at the same predictable dates $(T_n)_{n \ge 1}$. Allowing different predictable schedules requires only minor notational modifications. 
Let $(d_n)_{n \ge 1}$ denote the dividend payments at time $T_n$. The cumulative dividend process is  
$$D_t := \sum_{n\ge 1} \mathbf{1}_{\{T_n\le t\}}d_n,$$
where $\mathbf{1}_{A}$ is the indicator function of a set $A$ and we write 
$d_{t^-} = \sum_{n\ge 1} d_n \mathbf{1}_{\{T_n< t\}}$
for the most recent dividend level immediately before time $t$.

\subsection{Asset Dynamics}
The firm's asset value evolves continuously between  scheduled disclosure dates and may experience state-dependent jumps at announcement times. Let $0<T_1<T_2<\ldots$ be an increasing sequence of predictable stopping times representing publicly scheduled disclosure events. At each announcement date $T_n$, the jump is characterized by a random mark $Z_n = (\xi_n, \eta_n) \in \ccE$, $\ccE = \mathbb R^m \times \mathbb R^p$ with $m,p \in \mathbb N$, where $\xi_n$ determines the asset-value jump and $\eta_n$ represents the observation noise revealed through the public disclosure. The sequence $(Z_n)_{n\ge 1}$ is assumed i.i.d and independent of the Brownian motion driving the continuous dynamics. To describe the predictable disclosure events, define the marked point measure
\begin{align}\label{random measure}
    \psi(dt,dz) = \sum_{i=1}^\infty \delta_{(T_i, Z_i)}(dt,dz),
\end{align}
where $z = (\xi, \eta) \in\ccE$. The effect of a disclosure is specified through a measurable jump amplitude function $G: \mathbb R_{\ge 0} \times \mathbb R^p \times \mathbb R^m  \to \mathbb R,$ allowing the jump size to depend on the pre-announcement firm value, the prevailing observable information, and the disclosure mark. 
Consequently, the asset value dynamics under $\mathbb{Q}$ are given by
\begin{align}\label{eq:asset_dynamics}
dV_t = V_{t^-} \big(\mu_V dt + \sigma_V dW_t\big)- \kappa dD_t + \int_{\mathbb R^m} G(V_{t^-}, Y_{t^-}, \xi) \, \psi(dt, dz),
\end{align}
where $\mu_V \in \mathbb{R}$, $\sigma_V > 0$, and $\kappa \in [0,1]$ determines the financing mechanism for dividend payments as in \cite{frey2009pricing}. Equation \eqref{eq:asset_dynamics} implies that, at each announcement date $T_n$,
\begin{align}\label{announcement date}
    V_{T_n} = V_{T_n^-} + G(V_{T_{n}^{-}},Y_{T_{n}^{-}},\xi_n)-\kappa d_n,
\end{align}
while between announcement dates the asset value follows a  geometric Brownian motion. The specification of $G$ is intentionally flexible. It accommodates multiplicative disclosure shocks $G(v,y,\xi)=-\gamma(y,\xi)v$, additive shocks $G(v,y,\xi)=\gamma(y,\xi)$, as well as more general state-dependent parameterizations. In particular, its dependence on the observable signal $Y_{t^-}$ allows information shocks to vary depending on the current informational context. A representative specification is as follows:
$$G(v,y,\xi)=-\xi v\bigg(
1+\beta\frac{|y-\bar y|}{\bar y}\bigg),\qquad \beta\ge0,$$
where $\beta$ measures the sensitivity of disclosure shocks to deviations of the observable signal from a reference level $\bar y$. This specification generates state-dependent announcement effects while preserving the general formulation of the model.
\paragraph{Standing Assumption}
Throughout the paper, we work under the following standing assumptions.
\begin{assumption}\label{ass:standing}
    The model satisfies: 
    \begin{enumerate}
    \item[(i)] The Brownian motion $W$ and the sequence of disclosure marks $(Z_n)_{n\ge 1}$ are mutually independent under $\mathbb Q$, and the marks  are i.i.d.
    \item[(ii)] The jump amplitude function $G$ is Borel measurable and satisfies $G(v,y,\xi)\ge -(1-\epsilon)v$ for some $\epsilon \in (0,1)$ and every admissible $(v,y,\xi)$. Moreover, there exists a constant $c < \epsilon/\kappa$ (with the convention that this condition is void if $\kappa=0$) such that $d_n \le c V_{T_n^-}$ a.s. for every $n$ (ensure $V>0$, Lemma \ref{positivity of asset value}). 
     
    \item[(iii)] There exists $p>2$ such that
    $$\sup_{n\ge 1} E^\mathbb Q[|1+G(V_{T_{n}^{-}},Y_{T_{n}^{-}},\xi_n)/V_{T_{n}^{-}}|^p]<\infty.$$
    \item[(iv)] The predictable announcement dates satisfy, $ \#\{n: T_n \le t\}<\infty$ almost surely for every $t\ge 0$.
\end{enumerate}
\end{assumption}
Assumption \ref{ass:standing} guarantees that the asset value process is well defined and excludes arbitrage under the pricing measure $\mathbb{Q}$. The integrability condition are standard for jump-diffusion models and ensure that the filtering, pricing, and hedging problems considered in the sequel are well posed.

\paragraph{Liability dynamics and dividend policy}
Following Frey and Schmidt \cite{frey2009pricing}, the liability process evolves according to
\begin{align}\label{eq: liability}
    K_t = K_0 + (1-\kappa)D_t, \quad K_0 >0,
\end{align}
where $\kappa \in [0,1]$ determines the financing mechanism for dividend payments. Consequently, $\Delta K_{T_n}=(1-\kappa)d_n$, so that liabilities increase only when dividends are financed through debt issuance. The extreme cases $\kappa=1$ and $\kappa=0$ corresponds to pure asset financing and pure debt financing, respectively. Default occurs whenever $V_t \le K_t$, and, in particular, at an announcement date $T_n$
$$V_{T_n} \le K_{T_n^-}+(1-\kappa)d_n.$$
Dividend payment are assumed to provide noisy information about the firm's  financial condition. Their conditional distribution, given the pre-announcement state, is described by transition kernel
\begin{align}\label{kernel}
    \nu_d (dx|d_{n-1},v,y) = \nu_d(x|d_{n-1},v,y)dx,
\end{align}
where $v= V_{T_{n}^{-}}$ and $y=Y_{T_{n}^{-}}$. We assume that $x \to \nu_d(x|d_{n-1},v,y)$ is jointly measurable and continuous in $v$ for almost every $x$, ensuring that the Bayesian filtering update developed in Section \ref{filtering sec} is well defined. For the numerical examples, we employ the autoregressive specification
\begin{align}\label{autoregressive}
    d_n = \lambda d_{n-1} + (1-\lambda)\delta_n V_{T_{n}^{-}}
\end{align}
with $\delta_n$ i.i.d sequence of random payout proportions. 

\subsection{Market information and Investor filtration}\label{subsec:marketinfo}

Investors do not observe the firm's asset value $V$ directly. Instead, they infer the firm's financial condition from publicly available information consisting
\begin{enumerate}
    \item \emph{Default information.} The investors observe whether default has occurred. Equivalently, they observe the default indicator process $I_t := \mathbf{1}_{\tau \le t}$.
    \item \emph{Dividend information.} Investors observe the cumulative dividend process $D$.
    \item \emph{Liability information.} Since the initial liability level $K_0$ is publicly known and the liability dynamics are deterministic conditional on dividend payments, the liability process $(K_t)_{t\ge 0}$ is observable.
    \item \emph{Disclosure signals.} Investors observe a càdlàg signal process $(Y_t)_{t\ge 0}$. 
\end{enumerate}
The disclosure process is assumed to be càdlàg and piecewise constant between announcement dates. At each predictable disclosure time $T_n$, it is updated according to
    $$\Delta Y_{T_n} = Y_{T_n} - Y_{T_{n}^{-}}=f(V_{T_{n}^{-}},Y_{T_{n}^{-}})+\eta_n$$
    where $f$ is a measurable observation function and $\eta_n$ denotes the observation noise component of the mark $Z_n = (\xi_n, \eta_n)$.
Accordingly, the investor filtration is defined by $\mathbb H =(\ccH_t)_{t\ge 0}$, where
\begin{align}\label{investor filtration}
    \ccH_t := \sigma \big(I_s, D_s, Y_s: 0\le s \le t\big)\vee \ccN,
\end{align}
with $\ccN$ denotes the collection of $\mathbb Q$-null sets. By construction, $\ccH_t \subset \ccF_t,\quad t\ge 0$, and the latent asset value process $V$ is generally not adapted to $\mathbb H$. Since the announcement dates $(T_n)_{n\ge 1}$ are publicly scheduled, they are assumed to be predictable stopping times with respect to both $\mathbb{F}$ and $\mathbb H$. Consequently, investors know the timing of future disclosures but not their informational content before realization. 

\subsection{The full information state process} 
Under the full information filtration $\mathbb F$, the firm's  asset value is directly observable. The complete state vector is therefore $(V_t, K_t,Y_t, d_{t^-})$, where the components denote the asset value, liability level, public disclosure, and most recently observed dividend, respectively. Combining Equations \eqref{eq:asset_dynamics} and \eqref{eq: liability}, and the disclosure dynamics yields the following hybrid jump-diffusion system 
\begin{align}\label{sys:full_info}
    \begin{aligned}
        dV_t &= V_{t^-} \, \big(\mu_V \, dt + \sigma_V \, dW_t\big)- \kappa dD_t + \int_{\mathbb R^m} G(V_{t^-},Y_{t^-}, \xi) \, \psi(dt, dz), \\
        dK_t &= (1-\kappa)D_t,\\
        dY_t &= \int_{\mathbb R^{m+p}} \big( f(V_{t^-},Y_{t^-}) + \eta\big) \, \psi(dt,dz),\\
        dD_t &= \sum_{n \ge 1}d_n \delta_{T_n}(dt).
    \end{aligned}
\end{align}
Between announcement date, the state process evolves continuously according to the diffusion dynamics of the asset value, while $K$, $Y$, and $D$ remain constant. At each predictable disclosure time $T_n$, the state variables are updated simultaneously according to
$$V_{T_n}=V_{T_n^-} + G(V_{T_n^-},Y_{T_{n-1}},\xi)-\kappa d_n, \qquad K_{T_n}=K_{T_n^-} + (1-\kappa)d_n,$$
$$Y_{T_n} =Y_{T_{n}^{-}} + f(V_{T_{n}^{-}},Y_{T_{n}^{-}})+\eta_n$$
The process $(V_t, K_t, d_{t^-}, Y_t)_{t \ge 0}$ is Markov with respect to the full information filtration $\mathbb F$ under the standing assumption. This Markov structure forms the basis for the nonlinear filtering problem developed in Section \ref{filtering sec}, where only the subfiltration $\mathbb{H}$ generated by publicly observable quantities is available. 

\begin{lemma}\label{positivity of asset value}
    Suppose Assumption \ref{ass:standing} holds and $V_0 >0$. Then $V_t>0$ $\mathbb Q$-a.s. for all $t \ge 0$.
\end{lemma}
\begin{proof}
    Between predictable announcement dates, the process $V$ evolves as a geometric Brownian motion and therefore remains strictly positive almost surely. 
    
    At each announcement date $T_n$
    $$V_{T_n}=V_{T_n^-} + G(V_{T_n^-},Y_{T_{n-1}},\xi)-\kappa d_n \ge \epsilon V_{T_n^-} - d_n \ge (\epsilon-\kappa c) V_{T_n^-}>0.$$
    The result follows inductively over the sequence of predictable announcement dates.
\end{proof}

\section{Pricing Under Full Information}\label{sec: full information}
Under the full information filtration, investors observe the complete state vector $(V_t, K_t, d_{t^-}, Y_t)$ together with all scheduled disclosures and announcement realizations. Under the risk-neutral measure $\mathbb Q$, the resulting state process is Markov. Consequently, prices of default-sensitive securities admit Markovian representation in terms of the current state variables. These full-information pricing formulas serve as the benchmark for the partial information valuation developed in Section \ref{sec: valuation}. 

\subsection{Valuation of Corporate Securities}
\subsubsection{Equity Valuation} \label{sec: valuation}
The pre-default value of the firm's equity under full information is defined as the expected discounted dividend stream under $\mathbb Q$:
\begin{align}\label{temp382}
    S_t := \mathbf{1}_{\{\tau > t\}} \, E^\mathbb Q\left(\int_t^\tau e^{-r(s-t)} dD_s \,\Big|\, \ccF_t\right).
\end{align}
Using the Markov property, there exists a measurable function $S: [0,+\infty) \times \mathbb R_{\ge 0}^3 \times \mathbb R^p\to \mathbb R_{\ge 0}$ such that 
\begin{align}\label{eq: security price}
    S_t = \mathbf{1}_{\tau > t} \, S(t, V_t, K_t, d_t, Y_t).
\end{align}

\subsubsection{Default Probabilities and Defaultable Debt}
For $0\le t \le T$, define the conditional survival probability 
$P(t,T):=\mathbb Q(\tau > T \mid \ccF_t).$
By the Markov property,
\begin{align}\label{survival markov}
    P(t,T) = \mathbf{1}_{\tau > t}\mathbb Q\left(\inf_{u \in (t,T]} (V_u - K_u) > 0 \,\Big|\, V_t, K_t, d_{t^-},Y_t\right)
    &=: \mathbf{1}_{\tau > t}p(t, T, V_t, K_t, d_{t-},Y_t), 
\end{align}
for a measurable survival function $p:[0,+\infty) \times \mathbb R_+^3 \times \mathbb R^p \to [0,1].$ Equivalently, by denoting $F_\tau$ as the conditional distribution function of the default time, we have  $p(t, T, V_t, K_t, d_{t-},Y_t)=1 - F_\tau(t, T, V_t, K_t, d_t,Y_t).$

In contrast to classical diffusion based structural models, default may occur through two distinct mechanisms: \textit{continuous default}, when the diffusion component of the asset process crosses the liability barrier between announcement dates; \textit{Jump-induced default}, when a predictable disclosure shock pushes the firm below the default boundary at an announcement dates. 
More precisely, jump-induced default occurs whenever
    $$V_{T_n}=V_{T_{n}^{-}} + G(V_{T_{n}^{-}}, Y_{T_{n}^{-}}, \xi_n) - \kappa d_n \le K_{T_n}.$$
This hybrid structure implies that the default distribution generally contains both absolutely continuous and discrete components. Economically, this reflects the fact that credit deterioration may occur gradually through persistent financial weakness or abruptly following scheduled disclosure events such as earnings announcements, rating reviews, or covenant reports. 

\paragraph{Defaultable zero-coupon bonds.}
Consider a defaultable zero-coupon bond with maturity $T$ and fractional recovery rate $\delta \in [0,1]$.
Under the risk-neutral measure $\mathbb Q$, the bond price satisfies
\begin{align} \label{bond price}
    B^\delta(t, T) = \mathbf{1}_{\tau > t} \, e^{-r(T-t)} \big[p(t, T, V_t, K_t, d_{t^-},Y_t) + \delta (1-p(t, T, V_t, K_t, d_{t^-},Y_t))\big]. 
\end{align}
In the zero-recovery case $\delta=0$, $B(t, T) = \mathbf{1}_{\{\tau > t\}} \, e^{-r(T-t)} p(t, T, V_t, K_t, d_{t^-},Y_t)$. 

\paragraph{Credit Default Swaps.}
Following the market convention in \cite[Section 10.4.4]{mcneil2015quantitative}, the fair CDS spread  is defined as the premium rate for which the contract has zero value at inception. We consider payment dates $t<t_1<\cdots<t_N$, constant accrual period $\Delta =t_n - t_{n-1}$, deterministic loss-given-default parameter $\ell \in [0,1]$, and ignore accrued premium payments. We formalize this notion in the following proposition.
\begin{proposition}
    The fair CDS spread $x^*$ admits the representation
    \begin{align} \label{fair spread}
    x^* = \mathbf{1}_{\tau>t} \dfrac{\displaystyle \ell \int_t^{t_N} e^{-r(s-t)} F_\tau(t, ds, V_t, K_t, d_t,Y_t)}{\displaystyle \sum_{n:\, t_n >t} e^{-r(t_n - t)}\Delta \big(1 - F_\tau(t, t_n, V_t, K_t, d_t,Y_t)\big)}.
\end{align}
\end{proposition}
The representation in \eqref{fair spread} shows that predictable disclosure risk enters CDS premia through the conditional default distribution. Since default probabilities increase immediately before scheduled announcements, equilibrium CDS spreads reflect anticipated disclosure risk even before new information is released.

\paragraph{Credit Spreads}
The credit spread $s(t,T)$ for a bond with maturity $T$ is defined implicitly by $B^\delta(t, T) = e^{-(r + s(t,T))(T-t)}.$ For simplicity, set $r=0$. Then 
\begin{align}\label{temp447}
    s(t, T) = -\frac{1}{T-t} \log\big(p(t, T, V_t, K_t, d_t,Y_t) +\delta(1-p(t, T, V_t, K_t, d_t,Y_t))\big).
\end{align}
    In practice, the computation of $p(t,T,\cdots)$ is analytically intractable and must therefore rely on numerical methods. This difficulty stems from three intertwined features of the model. First, the default time is defined as a first-passage time of the asset process through a time-varying barrier, which makes survival probabilities inherently path-dependent. Second, the presence of state-dependent jumps at predictable dates introduces discontinuities in the dynamics and gives rise to jump-induced default. Third, the interaction between continuous diffusion component and the discrete predictable jumps prevents a decomposition of default risk into independent continuous and jump parts.

    As a consequence, standard closed-form techniques for structural models are no longer applicable. Instead, Monte Carlo simulation with suitable variance reduction methods such as importance sampling near the default barrier or efficient PDE-based schemes on time-space grids between jump dates provide natural and flexible computational approaches. We discuss these numerical methods in detail in Section \ref{sec:numerical}. The valuation formulas derived above assume that the value of the firm’s assets is directly observable. In practice, however, investors observe only the public information process described in Section \ref{subsec:marketinfo}. The following section therefore replaces the latent state variables with their nonlinear filter and derives valuation formulas under conditions of partial information.

\section{Incomplete Information}\label{sec : the incomplete market} 
Under the investor filtration $\mathbb H$, the firm's asset value is not directly observable. Consequently, the valuation formulas derived in Section \ref{sec: full information} cannot be evaluated from market information alone. This section characterizes the conditional distribution of the latent firm value given the observable filtration and develops pricing formulas under partial information. The analysis is based on the nonlinear filtering framework with stochastic discontinuities developed in  \cite{schmidt2026nonlinear}.

 \subsection{The Filtering Problem} \label{filtering sec}
 Under the investor filtration $\mathbb H$, the liability process $K$, cumulative dividend process $D$, and the public information process $Y$ are observable, whereas the asset value $V$ remains latent. Consequently, valuation reduces to computing conditional expectations with respect to the posterior distribution of $V$. 
For every bounded measurable test function $\varphi \in C_b^2(\mathbb R_{\ge 0})$, define
\begin{align}\label{eq:filtering}
    \pi_t(\varphi) := E^\mathbb Q[\varphi(V_t) | \ccH_t], \quad t \ge 0.
\end{align}
The measure-valued process $\pi=(\pi_t)_{t\ge 0}$ completely characterizes the investor's posterior beliefs about the hidden firm value.
The conditional distribution $\pi_t$ can be equivalently characterized by the family of stochastic processes $\{\pi_t(\varphi): \varphi \in C_b^{2}(\mathbb R_{\ge 0})\}$, where we allow for test functions that may depend on time.

\subsection{The Signal-Observation System}
Since $V_t>0$ before default by Lemma \ref{positivity of asset value}, introduce the logarithmic state variable $X_t = \log V_t$. The transformation converts the multiplicative diffusion dynamics of $V$ into an additive stochastic differential equation and places the model within the filtering framework of \cite{schmidt2026nonlinear}.

\paragraph{Log-transformation and signal dynamics.} 
Between two consecutive announcement dates $T_n$ and $T_{n+1}$, the firm value process satisfies 
$dV_t = V_t(\mu_V dt + \sigma_V dW_t)$. Applying It\^o's formula to $X_t = \log V_t$ yields 
\begin{align*}
    dX_t = d\log V_t = \left(\mu_V - \dfrac{\sigma_V^2}{2}\right) dt + \sigma_V dW_t.
\end{align*}
Define the coefficients
\begin{align*}
    a(X_t) := \mu_V - \dfrac{\sigma_V^2}{2}, \quad b(X_t) := \sigma_V.
\end{align*}
At predictable announcement times $T_n$, the asset value undergoes the jump  
$V_{T_n} = V_{T_n^{-}}+ G(V_{T_n^{-}}, Y_{T_n^{-}}, \xi_n) - \kappa d_n$. Consequently,
\begin{align*}
    X_{T_n} = X_{T_n^{-}} + \log\bigg(\dfrac{V_{T_n}}{V_{T_n^{-}}}\bigg) = X_{T_n^{-}} + \log\left(1 + \frac{G(V_{T_n^{-}}, Y_{T_n^{-}}, \xi_n) - \kappa d_n}{V_{T_n^{-}}}\right).
\end{align*}
We therefore define the jump coefficient:
\begin{align}\label{temp540}
    \tilde c(x, \xi) := c(x,Y_{T_n^{-}}, \xi)= \log\left(1 + \frac{G(e^x, Y_{T_n^{-}}, \xi) - \kappa d_{t^-}}{e^x}\right).
\end{align}

The signal process then satisfies
$$dX_t = a(X_t)dt + b(X_t)dW_t+\int_{\mathbb R^{m+p}} \tilde c(X_{t^{-}},\xi)\psi(dt,dz),$$
where $\psi$ is the predictable jump measure defined in \eqref{random measure} and $\tilde c: \mathbb R \times \mathbb R^m \to \mathbb R$.

\paragraph{Observation Process}
The observable information process $Y$ is piecewise constant and evolves only through jumps at the predictable announcement dates: 
\begin{align}
    Y_t = \sum_{n\ge 1} \mathbf{1}_{T_n \le t} \Delta Y_{T_n}, \quad \text{where} \quad \Delta Y_{T_n} = f(V_{T_n^{-}}, Y_{T_{n-1}}) + \eta_n.
\end{align}
Equivalently, $Y_{T_n}=Y_{T_{n-1}}+f(V_{T_n^{-}}, Y_{T_{n-1}}) + \eta_n$. 
Introducing the transformed observation function $\tilde{f}(x, y) := f(e^x, y),$ the observation equation becomes
$$dY_t = \displaystyle\int_{\mathbb R^{m+p}} \big(\tilde{f}(X_{t^-}, Y_{t^-}) + \eta\big) \, \psi(dt, dz),$$
with a measurable observation function $\tilde f: \mathbb R \times \mathbb R^p \to \mathbb R^p$.
Hence the partially observed system admits the compact representation
\begin{align}\label{eq:log_system}
    \begin{cases}
        dX_t = a(X_t) \, dt + b(X_t) \, dW_t + \displaystyle\int_{\mathbb R^{m+p}} \tilde c(X_{t^-}, \xi) \, \psi(dt, dz), \\[0.5em]
        dY_t = \displaystyle\int_{\mathbb R^{m+p}} \big(\tilde{f}(X_{t^-}, Y_{t^-}) + \eta\big) \, \psi(dt, dz),
    \end{cases}
\end{align}
for $t \ge 0$, with $X_0 = \log V_0$ and $Y_0 = 0$. 
The system \eqref{eq:log_system} is a partially observed diffusion with predictable stochastic discontinuities and therefore falls within the general filtering framework developed in \cite{schmidt2026nonlinear}.

\subsection{Specialization of the Kushner-Stratonovich Equation to Structural Credit Risk}\label{KS for credit risk}
The partially observed system \eqref{eq:log_system} satisfies the assumptions of \cite{schmidt2026nonlinear}. Specializing their general result yields the following Kushner-Stratonovich equation for the present structural credit-risk model.

We introduce the extended measurable space $(\tilde \Omega, \tilde {\ccP})=(\Omega \times \mathbb R_{\ge 0} \times \mathbb R^p, \ccP \otimes \ccB(\mathbb R^p))$, where $\tilde{\ccP}$ denotes the predictable sigma-algebra generated by the left-continuous adapted processes. 

\paragraph{Operators and Notation}
For $\varphi \in C_b^2(\mathbb R;\mathbb R)$, denote by $\ccL$ the generator of the continuous part of $X$, i.e.
\begin{align*}
    \ccL \varphi(x) = a \dfrac{\partial \varphi}{\partial x}(x) + \dfrac{b^2}{2} \dfrac{\partial^2 \varphi}{\partial x^2}(x) = \left(\mu_V - \dfrac{\sigma_V^2}{2}\right) \dfrac{\partial \varphi}{\partial x}(x) + \dfrac{\sigma_V^2}{2} \dfrac{\partial^2 \varphi}{\partial x^2}(x),
\end{align*}
and furthermore, for the jump part
\begin{align*}
    \ccA \varphi(x) = \int_{\mathbb R^{m+p}} \left[\varphi(x + \tilde c(x, \xi)) - \varphi(x)\right] F_Z(d\xi , d\eta),
\end{align*}
where $F_Z$ denotes the distribution of $Z_1$.
As in \cite{schmidt2026nonlinear}, we denote by $\mu$ the measure associated to the observation $Y$ and is given by
\begin{align*}
    \mu(dt, dy) = \sum_{k\ge 1} \delta_{(T_k, \Delta Y_{T_k})}(dt, dy).
\end{align*}
Its $\mathbb H$-compensator proven in \cite[Lemma 2.1]{schmidt2026nonlinear}, is given by
\begin{align*}
    \nu(dt, dy) = \sum_{k=1}^\infty \delta_{T_k}(dt) \, F^k(dy),
\end{align*}
where $F^k(A) = E^\mathbb Q\left[F(A - \tilde{f}(X_{T_{k^-}}, Y_{T_{k^-}})) \,\Big|\, \ccH_{T_{k^-}}\right], \quad A \in \ccB(\mathbb R^p)$
is the regular conditional distribution of $\Delta Y_{T_n}$ given $\ccH_{T_n^{-}}$, with $F:=\mathbb Q(\eta_1 \in \cdot)$ the distribution of $\eta_1$. 

\begin{remark}
    Since $X_{T_{k^-}}$ is not $\ccH_{T_k^{-}}$-measurable, the compensator $F^k(A)$ is computed by integrating over the filtered distribution:
    $$F^k(A) = \int_{\mathbb R} F(A - \tilde{f}(x, Y_{T_{k^-}}))\pi_{T_k^-}(dx),$$
    which is $\ccH_{T_k^{-}}$-measurable since $\pi_{T_k^-}$ is. Hence $\nu$ is a random (filtered) compensator. The compensated random measure $\tilde{\mu}(dt, dy) = \mu(dt, dy) - \nu(dt, dy)$ is an $\mathbb H$-martingale measure by the compensator definition \cite[Lemma 2.1]{schmidt2026nonlinear}.
\end{remark}

\subsubsection{Main Filtering Result}
\begin{assumption} \label{assumption 1}
Assume that for all $t \ge 0$, the functions $a,b,\tilde c,\tilde f$ from \eqref{eq:log_system} satisfy
\begin{align*}
        E\bigg[\displaystyle\int_0^t \int_{\mathbb R^p} \nu(ds,dy) \bigg]< \infty
        ,\quad 
        E\bigg[\displaystyle \int_0^t \|b(X_s)\|^2\,ds \bigg] < \infty,\quad 
        E\bigg[\displaystyle \int_0^t |a(X_s)|\,ds \bigg] < \infty \\
        E\bigg[\displaystyle\int_0^t \int_{\mathbb R^p} |\tilde c(X_{s-})| \,\nu(ds,dy) \bigg]< \infty,\quad E\bigg[\displaystyle\int_0^t \int_{\mathbb R^p} |\tilde f(X_{s-},y)| \, \nu(ds,dy) \bigg]< \infty.
\end{align*}
\end{assumption}
Let $A_t = \sum_{n\ge 1} \mathbf{1}_{T_n \le t}$ denotes the predictable counting process of jump times. The innovation functional captures how beliefs are updated when observation increment $\Delta Y_t = y$ is revealed
\begin{align}\label{eq:form of S}
    S(\varphi)(t, y) = E^\mathbb Q[\varphi(X_t) - \varphi(X_{t^-}) \mid \ccH_{t^-}, \Delta Y_t = y].
\end{align}
Before state the result, we recall that for a random measure $\mu$ and a function $S$, we define the stochastic process $S*\mu$ through
    $$(S*\mu)_t(\omega)=\int_{[0,t]\times\mathbb R^p} S(\omega;s,y)\mu(\omega;ds,dy),$$ if the integral w.r.t $|S|$ is finite, and $+\infty$ otherwise. we then characterizes the default compensator under incomplete information.

\begin{theorem}\label{temp616}
    Under Assumption \ref{assumption 1}, the filter $(\pi_t)$ defined in \eqref{eq:filtering} with $V_t=e^{X_t}$ satisfies, for every $\varphi \in C_b^{2}(\mathbb R;\mathbb R)$,
\begin{align}\label{eq:ks_credit}
    \pi_t(\varphi) = \pi_0(\varphi) + \int_0^t \pi_s(\ccL \varphi) \, ds + \int_0^t \pi_{s^-}(\ccA \varphi ) \, dA_s + (S(\varphi)*\tilde \mu)_t,\quad   t\ge 0.
\end{align}

\end{theorem}

\begin{proof}
The proof follows from \cite[Proof of Theorem 2.5]{schmidt2026nonlinear}, translated to the credit risk system \eqref{eq:log_system}.
    For $\varphi \in C_b^{2}(\mathbb R)$, define 
    \begin{equation} \label{Mat1}
    M^\varphi_t = \varphi (X_{t}) - \varphi (X_0) - \int_0^t \ccL\varphi(X_{s})ds - 
    \int_0^t \ccA \varphi(X_{s-}) dA_s.
\end{equation} 
By \cite[Proposition 4.2]{schmidt2026nonlinear}, $M^\varphi$ is an $\mathbb F$-martingale. Taking conditional expectations on Equation \eqref{Mat1} with respect to $\ccH_t$, using Lemmas 4.3-4.5 of \cite{schmidt2026nonlinear}, yields
    \begin{align}\label{cred 1}
    \pi_t(\varphi) = \pi_0(\varphi) + \int_0^t \pi_s(\ccL \varphi) \, ds + \int_0^t \pi_{s^-}(\ccA \varphi ) \, dA_s + M_t,\quad   t\ge 0.
\end{align}
where $M$ is an $\mathbb H$-martingale. 
Since the observation filtration is generated by the jump measure $\mu$, the martingale representation theorem for integer-valued random measures \cite[Proposition 4.1]{schmidt2026nonlinear} implies the existence of a predictable process $S$ such that $M_t= (S*\tilde{\mu})_t$ on $\{\tau >t\}$ ($I_t = \mathbf{1}_{\tau \le t}=0$ and the filtration is generated only by the observation $Y$ and the survival information). Finally, identifying the jump part and follow the same analysis as in \cite{schmidt2026nonlinear} gives,
$$S:= S(\varphi)(t,y)=E^\mathbb Q[\Delta \varphi(X_t) |\ccH_{t^-},\Delta Y_{t}=y].$$
Substituting this representation into the previous equation proves \eqref{eq:ks_credit}.
\end{proof}
\begin{remark}
    The filtering equation \eqref{eq:ks_credit} is the continuous-state analogue of the discrete recursive filter obtained by \cite{frey2009pricing} through Markov chain approximation methods which basically is discretizing both time and the state space, then recursively computing
    \begin{align*}
    q_j(k) = \mathbb Q(V_k = m_j(k) | \ccH_k), \quad j=1,\ldots,|M|
\end{align*}
where $\{m_1(k), \ldots m_{|M|}(k)\}$ is a finite grid and $V_k$ is a discrete-time approximation. In the present framework, the innovation term $(S(\varphi) * \tilde \mu)_t$ captures the Bayesian update generated by predictable information releases at the announcement dates $(T_n)$. 
\end{remark}
\subsubsection{Economic Interpretation of the Filtering Equation}
The Kushner-Stratonovich equation \eqref{eq:ks_credit} separates naturally into three information channels.
\paragraph{Absence of default}
Between disclosure dates, the survival of the firm is informative because firms with lower asset values are more likely to default. Consequently, conditional on no default, the posterior distribution gradually shifts toward stronger financial states. This continuous learning mechanism coincides with that in partially observed structural credit-risk models such as Frey and Schmidt \cite{frey2009pricing}.

\paragraph{Predictable disclosures}
At each scheduled announcement time, investors incorporate newly released information through a Bayesian update driven by the disclosure mark. Unlike classical structural models with totally inaccessible jump times, the timing of these information arrivals is known in advance. As a result, market participants anticipate periods of heightened information risk before announcements, while uncertainty is resolved immediately after the disclosure.

\paragraph{Innovation process}
The innovation term measures the discrepancy between the realized disclosure and its conditional prediction under the current posterior distribution. Large innovations induce substantial revisions of investors' beliefs, whereas small innovations produce only moderate updates. This mechanism provides the dynamic link between public disclosures and market valuations under incomplete information.

\subsection{Connection to Pricing under Incomplete Information}

The full information valuation formulas of Section \ref{sec: full information} immediately yield their partial-information counterparts by conditioning on the investor filtration.
The market value of equity therefore becomes 
\begin{align*}
    S_t^{\ccH} = \mathbf{1}_{\tau > t} \, E^\mathbb Q[S(t, V_t, K_t, d_t,Y_t) | \ccH_t].
\end{align*}
Since $K_t$, $d_t$ and $Y_t$ are $\ccH_t$-measurable, we obtain
\begin{align}\label{equity under incomplete information}
    S_t^{\ccH} = \mathbf{1}_{\tau > t} \int_{\mathbb R_{\ge 0}} S(t, v, K_t, d_t,Y_t) \, \pi_t(dv),
\end{align}
Similarly, the conditional survival probability satisfies
\begin{align*}
    \mathbb Q(\tau > T | \ccH_t) = \mathbf{1}_{\tau > t} \int_{\mathbb R_{\ge 0}} p(t, T, v, K_t, d_t,Y_t) \, \pi_t(dv),
\end{align*}
where $p(t,T,v,K,d,y)=\mathbb Q(\tau >T \mid V_t=v, K_t=K, d_t=d,Y_t=y)$ denotes the full information survival probability introduced in Section \ref{sec: full information}. Consequently, defaultable bond prices under incomplete information admit the representation
\begin{align*}
    B^\delta(t, T) = \mathbf{1}_{\{\tau > t\}}e^{-r(T-t)} \int_{\mathbb R_{\ge 0}}  \left[\delta + (1-\delta) p(t, T, v, K_t, d_t,Y_t)\right] \, \pi_t(dv).
\end{align*}
Likewise, the fair CDS spread becomes:
\begin{align*}
    x^*_{\ccH} = \dfrac{\displaystyle \ell \int_t^{t_N} e^{-r(s-t)} \int_{\mathbb R_{\ge 0}} \mathbb Q(\tau \in ds\mid V_t=v, \ccF_t) \pi_t(dv)}{\displaystyle  \sum_{n:t_n>t} e^{-r(t_n - t)}\Delta \int_{\mathbb R_{\ge 0}} p(t, t_n, v, K_t, d_t,Y_t) \pi_t(dv)},
\end{align*}

Hence, the filtering problem of determining the conditional distribution $\pi_t$ becomes central to the valuation of all corporate securities under incomplete information.

\subsubsection{Particle Filter Approximation}
The Kushner-Stratonovich equation is an infinite-dimensional stochastic evolution equation and generally admits no closed-form solution. We therefore approximate the 
posterior distribution by a sequential Monte Carlo (particle filter) method. Specifically,
$$\pi_t^N=\dfrac{1}{N}\sum_{i=1}^N w_t^{(i)} \delta_{V_t^{(i)}},$$
where $V_t^{(i)}$ denotes the $i$th particle and $\{w_t^{(i)}\}_{i=1}^N$ are normalized weights satisfying $\sum_{i=1}^N w_t^{(i)}=1$. The algorithm alternates between propagation, Bayesian updating at scheduled disclosure dates, and resampling whenever particle degeneracy becomes significant. The following conditions are imposed throughout this subsection.

\begin{assumption}\label{filter condition}
    \begin{enumerate}
        \item The functions $G(v,y,\xi)$ and $f(v,y)$ are globally Lipschitz continuous in $v$.
        \item There exists $p\ge 4$ such that  $$\sup_{n\ge 1} E^\mathbb Q[|G(V_{T_n^-},Y_{T_n^-},\xi_n)|^p]<\infty,\quad \sup_{n\ge 1} E^\mathbb Q[|f(V_{T_n^-},Y_{T_n^-})+\eta_n|^p]<\infty.$$
        \item The observation noise admits a density $p_\eta$ satisfying for some constants $C_\eta$,
        $$\sup_{\eta \in \mathbb R^p}p_\eta(\eta)\le C_\eta<\infty.$$
        \item The conditional dividend density satisfies $\sup_{v,y}\nu_d(d_m|d_{m-1},v,y)\le C_d < \infty,$
        for some constant $C_d>0$, for every announcement date.
        \item The initial condition satisfies $E^\mathbb Q[|V_0|^p]<\infty$.
    \end{enumerate}
\end{assumption}

\begin{proposition}\label{filter convergence}
    Let Assumption \ref{filter condition} hold. For every $\varphi \in C_b^{2}(\mathbb R)$ and every fixed $t \ge 0$, the particle approximation satisfies
    \begin{align*}
        \pi_t^N(\varphi) \longrightarrow \pi_t(\varphi) \quad \text{a.s. as}\, N \to \infty, \, \forall t\ge 0.
    \end{align*}
\end{proposition}
The proof is based on standard convergence results for sequential Monte Carlo methods applied to jump-diffusion filtering problems and  is given to Appendix \ref{Appendix A}. 

The full-information and partial information pricing frameworks are linked through the posterior distribution generated by the nonlinear filter. Under complete information, security values depend directly on the latent state variables $(V_t, K_t, d_t, Y_t)$. Under partial information, these state variables are replaced by their conditional distribution given the investor filtration. Consequently, all valuation formulas derived in Section \ref{sec: full information} remain valid after taking conditional expectations with respect to the filter. 

The particle approximation provides an efficient numerical implementation of this procedure by recursively updating the posterior distribution as new disclosures arrive. Since the latent asset value cannot be perfectly observed, the same filtered distribution also forms the basis for the quadratic hedging strategies developed in Section \ref{sec 5: Hedging}, thereby providing a unified framework for pricing and risk management under incomplete information.

\section{Hedging Under Incomplete Information With Predictable Jumps}\label{sec 5: Hedging}
We now study hedging of default-sensitive claims under the investor filtration $\mathbb{H}$. Since the firm's asset value is not directly observable and disclosure shocks are revealed only at scheduled announcement dates, the market is incomplete and perfect replication is generally impossible. We therefore adopt the local risk-minimization framework of \cite{schweizer1991option,moller2001risk} extending the predictable jump setting of \cite{fontana2024term} to the present structural credit risk models.

\subsection{Hedging Framework and Market Incompleteness }
We work in the incomplete market setting introduced in Section \ref{sec : the incomplete market}, where market incompleteness arises from partial information. Trading takes place under the restricted investor filtration $\ccH =(\ccH_t)_{t \ge 0}$, generated by observable dividend payments, noisy public announcements, and default information, but excluding the true firm value process $V$. All quantities are expressed in discounted units. The discounted asset price is $\tilde S_t^\ccH = e^{-rt}S_t^\ccH,$ and the money market account is taken as num\'eraire. In particular, portfolio and trading strategies are formulated directly in discounted terms.

\begin{definition}\label{def:trading strategy}
    A trading strategy is a pair $\varphi = (\zeta,\phi)$ where $\zeta = (\zeta_t)_{t \ge 0}$ is an $\mathbb{H}$-predictable process, $\phi = (\phi)_{t \ge 0}$ is adapted and the discounted wealth process is $$\tilde V_t (\varphi):=e^{-rt}(\zeta_t S_t^\ccH+\phi_t ),$$ and has right-continuous paths satisfying  $V_t (\varphi) \in L^2(\ccH_t,\mathbb{Q})$ for all $t \in [0,T].$
\end{definition}
 We denote by $\Theta(\ccH)$ the class of  $\ccH$-admissible trading strategies. For any claim $H$, we write $\tilde H:= e^{-rT}H.$
\begin{definition}
    Let $H\in L^2(\ccH_T,\mathbb{Q})$ be a contingent claim. A  strategy $\varphi \in \Theta(\ccH)$ is \emph{admissible} for $H$ if 
    $\tilde V_T(\varphi) = \tilde H \quad \mathbb{Q}\text{-a.s.}.$
\end{definition}

\begin{definition}\label{Mean self-financing}
    An admissible strategy $\varphi = (\zeta,\phi) \in \Theta(\ccH)$ is said to be \emph{mean self-financing} if its costs process
    $$C_t(\varphi) := \tilde V_t (\varphi)-\int_0^t \zeta_u d\tilde S_u^\ccH \qquad \text{for all}\quad t \in [0,T]$$
    is an $(\mathbb{H},\mathbb{Q})$-martingale.
\end{definition}
\subsection{Local Risk-minimization and Galtchouk-Kunita-Watanabe Decomposition}
We adopt the local risk-minimization criterion of \cite{schweizer1991option}, which minimizes the conditional variance of future hedging costs.
\begin{definition}\label{LRM}
Let $H\in L^2(\ccH_T,\mathbb{Q})$. A mean self-financing strategy $\varphi^\ast \in \Theta(\ccH)$ is  \emph{locally risk-minimizing} if
    \begin{enumerate}
        \item[i)] $\tilde{V}_T(\varphi^\ast)=\tilde H_T$ a.s.,
        \item[ii)] for all $t \in [0,T]$ and all $\varphi \in \Theta(\ccH)$
        $$E^\mathbb{Q}\big[ (C_T(\varphi^*)-C_t(\varphi^*))^2|\ccH_t\big]\le E^\mathbb{Q}\big[ (C_T(\varphi)-C_t(\varphi))^2|\ccH_t\big]$$
    \end{enumerate}
\end{definition}

\subsubsection{Galtchouk-Kunita-Watanabe decomposition with predictable Jump}
The characterization of locally risk-minimizing strategies relies on  is the Galtchouk-Kunita-Watanabe (GKW) decomposition of the discounted claim with respect to the traded asset.

\begin{definition}
    Two $\mathbb{H}$-local martingales $M=(M_t)_{t \ge 0}$ and $N=(N_t)_{t \ge 0}$ are called \emph{strongly orthogonal}, written as $M \perp N$ if their product $MN$ is a $\mathbb{H}$-local martingale. Equivalently, $\langle M, N\rangle_t =0$ for all $t\ge 0$, where $\langle \cdot,\cdot \rangle$ denotes the predictable quadratic covariation.
\end{definition}
Under the investor filtration, the discounted equity price $\tilde S^\ccH$ is a square-integrable semmimartingale whose dynamics are driven by the innovation process associated with the nonlinear filter. Its canonical decomposition consists of a continuous martingale part together with jumps occurring only at the predictable announcement times: 
\begin{align*}
    \tilde S_t^\ccH = \tilde S_0^\ccH +  \tilde S_{t}^{\ccH,c} + \sum_{T_n \le t} \Delta\tilde S_{T_n}^\ccH.
\end{align*}
The continuous component is generated by the innovation $\mathbb{H}$-Brownian motion, whereas the jump component reflects Bayesian updates at scheduled disclosures.
We emphasize that Definition \ref{LRM} is equivalent to the requirement that the associated cost process $C=(C_t)_{t\ge 0}$ satisfies
$$\text{$C$ is a square-integrable $\mathbb{H}$-martingale strongly orthogonal to $\tilde S^\ccH$}.$$ 

\begin{assumption} \label{assumption 2}
    The discounted price $\tilde S^\ccH$ is a square-integrable $\mathbb{H}$-martingale, whose jumps occur only at predictable announcement times and satisfies $$E^\mathbb{Q} \big[\sum_{T_n \le t} (\Delta\tilde S_{T_n}^\ccH)^2\big]<\infty$$
\end{assumption}

\begin{lemma}
    Under the investor filtration $\mathbb{H}$, every square-integrable martingale in the subspace generated by predictable announcement jumps admits the decomposition
    $$M_t = M_0 + M_t^c + \sum_{T_n \le t}\Delta M_{T_n},$$
    where $M^c$ is the continuous martingale part and the jumps occur only at the predictable announcement times $(T_n)$ and at the default time $\tau$.
\end{lemma}
In the present section, we focus on claims whose hedgeable risk is generated by the predictable announcement jumps.

\begin{theorem}\label{GKW Decomposition}
    Let Assumption \ref{assumption 2} holds and let $H_T \in L^2(\ccH_T, \mathbb{Q})$. Define
    $$\tilde H_t := E^\mathbb{Q}[e^{-rT}H_T| \ccH_t], \quad t \in [0,T].$$
    Then there exists a unique decomposition
    \begin{align}\label{GKW equation}
        \tilde H_t = \tilde H_0 + \sum_{T_n \le t} \zeta^j_{T_n}\Delta\tilde S_{T_n}^\ccH+L_t, \quad t \in [0,T],
    \end{align}
    where 
     $L=(L_t)_{t\ge 0}$ is an $\mathbb{H}$-martingale with $L_0=0$, strongly orthogonal to $\tilde S^\ccH$, and $\zeta^j = (\zeta^j_{T_n})_{n \ge 1}$ is a sequence of $\ccH_{T^-_n}$-measurable random variable with $\sum_{T_n \le t} (\zeta^j_{T_n})^2(\Delta\tilde S_{T_n}^\ccH)^2<\infty$ a.s.
\end{theorem}
\begin{proof}
    Since $\tilde H$ is a square-integrable $\mathbb{H}$-martingale, its jump component at the predictable announcement times admits the representation can be written as
    $$\tilde H_t = \tilde H_0 + \sum_{T_n \le t} \Delta\tilde H_{T_n}+L_t,$$
    where $L$ contains all martingale orthogonal to the predictable jump space generated by $\tilde S^\ccH$.

    At each predictable time $T_n$, 
    define the conditional $L^2$-projection coefficient
    \begin{align}\label{jump integrand}
        \zeta^j_{T_n} = \dfrac{E^\mathbb{Q}[\Delta \tilde H_{T_n}\Delta \tilde S_{T_n}^\ccH| \ccH_{T^-_n}]}{E^\mathbb{Q}[(\Delta \tilde S_{T_n}^\ccH)^2| \ccH_{T^-_n}]},
    \end{align}
    with the convention $\zeta_{T_n}^j=0$ if the denominator vanishes. Define the residual jump process
    $$\Delta L_{T_n}=\Delta\tilde H_{T_n}-\zeta_{T_n}^j\Delta\tilde S_{T_n}^\ccH,$$
    and set 
        $L_t = \sum_{T_n \le t}\Delta L_{T_n}, \quad t \in [0,T].$
    By construction, $E^\mathbb{Q}[\Delta L_{T_n}\Delta \tilde S_{T_n}^\ccH| \ccH_{T^-_n}]=0,$ for every $n$. Hence, 
    \begin{align*}
        \langle L, \tilde S^\ccH \rangle_t =\sum_{T_n \le t}E^\mathbb{Q}[\Delta L_{T_n}\Delta \tilde S_{T_n}^\ccH| \ccH_{T^-_n}]=0, 
    \end{align*}
    which implies the strong orthogonality $L \perp \tilde S^\ccH.$ Substituting the definition of $\Delta L_{T_n}$ yields
    $$\tilde H_t = \tilde H_0 + \sum_{T_n \le t} \zeta^j_{T_n} \Delta\tilde S_{T_n}+L_t,$$

    To prove uniqueness, suppose that two decompositions of the stated form exist. Subtracting them yields
    \begin{align*}
        0= \sum_{T_n \le t} (\zeta^j_{T_n}-\tilde\zeta^j_{T_n})\Delta \tilde S_{T_n}^\ccH+( L_t - \tilde L_t).
    \end{align*}
    Taking the predictable quadratic variation and using orthogonality gives:
    \begin{align*}
        \sum_{T_n \le t} (\zeta^j_{T_n}-\tilde\zeta^j_{T_n})^2(\Delta \tilde S_{T_n}^\ccH)^2+\langle L - \tilde L\rangle_t =0
    \end{align*}
    Thus, $\zeta_{T_n}^j = \tilde \zeta_{T_n}^j$ for all $n$ and $L=\tilde L$ a.s.
\end{proof}

\subsubsection{Locally risk-minimizing Strategy}
We show that the GKW decomposition \eqref{GKW equation} directly yields the locally risk-minimizing strategy.
\begin{theorem}\label{unique GKW}    Let $H_T \in L^2(\ccH_T,\mathbb{Q})$ and suppose its discounted value process admits the GKW decomposition \eqref{GKW equation}. Define the strategy $\varphi^* =(\zeta^*,\phi^*)$ by 
    \begin{align}
       \label{unique 1}  &\zeta^*_t= \zeta_{T_n}^j\, \text{if $t \in [T_n,T_{n+1})$},\quad \zeta^*_0= \zeta_{T_0}^j \\
      \label{unique 2}  &\phi_t^* := e^{rt}(\tilde H_t - \zeta_t^* \tilde S_t^\ccH).
    \end{align}
    Then $\varphi^* \in \Theta(\ccH)$ is the unique locally risk-minimizing strategy for $H_T$ in the sense of Definition \ref{LRM}. 
\end{theorem}

\begin{proof}
    We must verify the two condition in Definition \ref{LRM}.
    First, we need to show that $\tilde V_T(\varphi^*)=e^{-rT}H_T$ a.s. By definition of the discounted value process, we have
    $$\tilde V_T(\varphi^*)=e^{-rT}(\zeta_T^* S_T^\ccH+\phi_T^*)$$
    Substituting \eqref{unique 2}, we get
    \begin{align*}
        \tilde V_T(\varphi^*)=e^{-rT}(\zeta_T^* S_T^\ccH+e^{rt}(\tilde H_T - \zeta_T^* \tilde S_T^\ccH))
        = \zeta_T^*\tilde S_T^\ccH+ \tilde H_T - \zeta_T^* \tilde S_T^\ccH = \tilde H_T.
    \end{align*}
    Since $\tilde H_t = E^\mathbb{Q}[e^{-rT}H_T|\ccH_t]$ and at time $T$ we have $\ccH_T$-measurability of $H_T$, then
    $$\tilde H_T = E^\mathbb{Q}[e^{-rT}H_T|\ccH_T]=e^{-rT}H_T.$$
    Therefore $\tilde V_T(\varphi^*)=e^{-rT}H_T$, which verifies the terminal condition.

    Second, the discounted cost process of $\varphi^*$ is
    \begin{align*}
        C_t(\varphi^*) &= \tilde V_t (\varphi^*)-\int_0^t \zeta_u^* d\tilde S_u^\ccH
        = e^{-rt}(\zeta_t^* S_t^\ccH+\phi_t^*)-\int_0^t \zeta_u^* d\tilde S_u^\ccH.
    \end{align*}
    Using \eqref{unique 2}, we have $\phi_t^* := e^{rt}(\tilde H_t - \zeta_t^* \tilde S_t^\ccH)$, so as in the first step, $e^{-rt}(\zeta_t^* S_t^\ccH+\phi_t^*) = \tilde H_t$. Therefore,
    $$C_t(\varphi^*) = \tilde H_t-\int_0^t \zeta_u^* d\tilde S_u^\ccH.$$
    Using the GKW decomposition \eqref{GKW equation}, 
    $$\tilde H_t = \tilde H_0 + \sum_{T_n \le t} \zeta^j_{T_n}\Delta\tilde S_{T_n}^\ccH+L_t,$$
    since $\zeta^*_{T_n}=\zeta^j_{T_n}$ at jump times, and $\tilde S_t^\ccH = \sum_{T_n \le t} \Delta\tilde S_{T_n}^\ccH$, we have that
    $$\int_0^t \zeta^*_u d\tilde S^{\ccH}_u =  \sum_{T_n \le t} \zeta^j_{T_n}\Delta\tilde S_{T_n}^\ccH.$$
    Substituting into the expression of $C_t(\varphi^*)$, we have
    $$C_t(\varphi^*)=\tilde H_0  + \sum_{T_n \le t} \zeta^j_{T_n}\Delta\tilde S_{T_n}^\ccH+L_t - \sum_{T_n \le t} \zeta^j_{T_n}\Delta\tilde S_{T_n}^\ccH=\tilde H_0 +L_t, $$
    where $L$ is the martingale orthogonal to $\tilde S^\ccH$ from the GKW decomposition. Hence the cost process is a square-integrable martingale strongly orthogonal to $\tilde S^\ccH$.
    
    Now consider any other admissible strategy $\varphi \in \Theta(\ccH)$. Its cost process satisfies
    $$C_t(\varphi) = \tilde V_t (\varphi)-\int_0^t \zeta_u d\tilde S_u^\ccH.$$
    Since both $\tilde V_t (\varphi)$ and $\tilde V_t (\varphi^*)$ are $\ccH_t$-measurable and $\tilde V_t (\varphi^*)=\tilde H_t$, we can write 
    \begin{align*}
        C_t(\varphi) &= \tilde V_t (\varphi)-\int_0^t \zeta_u d\tilde S_u^\ccH\\
        &= (\tilde V_t (\varphi)- \tilde H_t) + \tilde H_t -\int_0^t \zeta_u d\tilde S_u^\ccH\\
        &= (\tilde V_t (\varphi)- \tilde H_t) + \int_0^t \zeta_u^* d\tilde S_u^\ccH-\int_0^t \zeta_u d\tilde S_u^\ccH+\tilde H_0 +L_t\\
        &= C_t(\varphi^*) + \int_0^t (\zeta_u^*-\zeta_u)d\tilde S_u^\ccH + (\tilde V_t (\varphi)- \tilde H_t)
    \end{align*}
    For a mean self-financing strategy, the last term form a martingale (see Definition \ref{Mean self-financing}). The increment of the cost process is:
    $$dC_t(\varphi) =dL_t + (\zeta_t^*-\zeta_t)d\tilde S_t^\ccH + dM_t,$$
    where $M=\tilde V (\varphi)- \tilde H$.
    Taking the conditional variance and using the fact that $L \perp \tilde S$,
    \begin{align*}
E^\mathbb{Q}\big[(dC_t(\varphi))^2|\ccH_t\big]=E^\mathbb{Q}\big[(dL_t)^2|\ccH_t\big]&+(\zeta_t^*-\zeta_t)^2E^\mathbb{Q}\big[(d\tilde S_t^\ccH)^2|\ccH_t\big]+E^\mathbb{Q}\big[(dM_t)^2|\ccH_t\big]\\&+2E^\mathbb{Q}\big[dL_t \cdot (\zeta_t^*-\zeta_t)d\tilde S_t^\ccH|\ccH_t\big],
    \end{align*}
    where the cross term vanishes by the orthogonality $L \perp \tilde S^\ccH$, the first and third terms are independent of $\zeta$, while the second term is minimized uniquely at $\zeta_t = \zeta_t^*$, therefore
    $$E^\mathbb Q\big[(dC_t(\varphi^*))^2|\ccH_t\big]\le E^\mathbb Q\big[(dC_t(\varphi))^2|\ccH_t\big]\quad \text{for all $\varphi \in \Theta(\ccH)$}.$$
    If another strategy $\tilde \varphi$ also satisfies both condition, then $C_t(\tilde \varphi)=L_t$ by the same argument. The GKW decomposition is unique (Theorem \ref{GKW Decomposition}), so the integrand $\zeta^j$ are uniquely determined, implying $\tilde \varphi = \varphi^*$ a.s.
\end{proof}
\subsection{Hedging Formula Under Incomplete Information}
In order to make Theorem \ref{unique GKW} operational, we derive the explicit expression for the predictable jump hedging coefficient using the filtering distribution $\pi_t$ established in Section \ref{filtering sec}. The following result provides the local risk-minimizing hedge ratio at Announcement dates. To this end, we denote by $F_\xi(d\xi)$ the probability measure (or law) of the random jump size variable $\xi_n$, which is the marginal distribution of the joint mark distribution $F_Z(d\xi,d\eta)$.

\begin{proposition}\label{hedge ratio}
    Let $H = h(V_T)$, where $h: \mathbb R_{\ge 0} \to \mathbb R$ is continuously differentiable and define $u(t,v):= E^\mathbb Q[h(V_T)|V_t=v]$. Then, on the event $\{\tau > T_n\}$, the predictable jump hedging coefficient at the announcement time $T_n$ is given by
    \begin{align}\label{form of zeta}
        \zeta^j_{T_n} = e^{-r(T-2T_n)}\dfrac{\text{Cov}_\pi(J^H,J^S)}{\text{Var}_\pi(J^S)},
    \end{align}
    where $J^H$ and $J^S$ denote the jump responses of the claim and the traded asset, respectively, and writing $\tilde G(v,y,\xi)=G(v,y,\xi)-\kappa d_{T_n^-}$, we set
    $$J^H(v,\xi):=u(T_n,v+\tilde G(v,Y_{T_n^-},\xi))-u(T_n,v),$$
    $$J^S(v,\xi):=S(T_n,v+\tilde G(v,Y_{T_n^-},\xi),K_{T_n},d_{T_n},Y_{T_n})-S(T_n,v,K_{T_n^-},d_{T_n^-},Y_{T_n^-}),$$
    and the filtered (co)variances are:
    $$\text{Cov}_\pi(J^H,J^S):=\displaystyle\int_{\mathbb R_{\ge 0}} \int_{\mathbb R^m} J^H(v,\xi)J^S(v,\xi)F_\xi(d\xi)\pi_{T_n^-}(dv),$$
    $$\text{Var}_\pi(J^S) := \int_{\mathbb R_{\ge 0}} \int_{\mathbb R^m} (J^S(v,\xi))^2 F_\xi(d\xi)\pi_{T_n^-}(dv).$$
    
\end{proposition}

\begin{proof}
   Fix a predictable announcement time $T_n$ and work on the event $\{\tau > T_n\}$, so that the firm has not defaulted prior to the information arrival.
    First, we have $\tilde H_t = E^\mathbb Q[e^{-rT}h(V_T)|\ccH_t]$, conditionally on the latent firm value $V_t =v$, the Markov property of $V$ implies 
    $$E^\mathbb Q[e^{-rT}h(V_T)|V_t=v]=e^{-r(T-t)}u(t,v).$$
    At the predictable time $T_n$, the firm value jump from $V_{T_n^-}$ to $V_{T_n}=V_{T_n^-}+G(V_{T_n^-},Y_{T_n^-},\xi)-\kappa d_{T_n^-}:=V_{T_n^-}+ \tilde G(V_{T_n^-},Y_{T_n^-},\xi)$. Therefore the jump of the discounted claim satisfies
    \begin{align*}
        \Delta \tilde H_{T_n}= e^{-r(T-T_n)}\big(u(T_n,V_{T_n^-}+ \tilde G(V_{T_n^-},Y_{T_n^-},\xi))-u(T_n,V_{T_n^-})  \big).
    \end{align*}
    Under incomplete information, $\tilde S_t^\ccH = e^{-rt}\int_{\mathbb R_{>0}} S(t,v,K_t,d_t,Y_t)\pi_t(dv)$. Define the claim jump
    $$J^H(v,\xi):=u(T_n,v+ \tilde G(v,Y_{T_n^-},\xi))-u(T_n,v),$$
    where $\tilde G(v,y,\xi)=G(v,y,\xi)-\kappa d_{T_n^-}$. and let $J^S(v,\xi)$ denote the corresponding equity jump function. Using the projection formula from Theorem \ref{GKW Decomposition}, the predictable hedging coefficient is 
    \begin{align*}
        \zeta^j_{T_n} = \dfrac{E[\Delta \tilde{H}_{T_n} \Delta \tilde S_{T_n}^\ccH |\ccH_{T_n^-} ]}{E[(\Delta \tilde{S}_{T_n}^\ccH)^2 |\ccH_{T_n^-} ]}.
    \end{align*}
    Conditioning further on the filtered state distribution $\pi_{T_n^-}$ yields
    $$E[\Delta \tilde H_{T_n} \Delta \tilde S_{T_n}^\ccH |\ccH_{T_n^-} ]=e^{-rT}\displaystyle\int_{\mathbb R_{>0}}\int_{\mathbb R^m}J^H(v,\xi)J^S(v,\xi)F_\xi(d\xi)\pi_{T_n^-}(dv),$$
    and similarly
    $$E[(\Delta \tilde{S}_{T_n}^\ccH)^2 |\ccH_{T_n^-} ] = e^{-2rT_n}\displaystyle\int_{\mathbb R_{>0}}\int_{\mathbb R^m}(J^S(v,\xi))^2F_\xi(d\xi)\pi_{T_n^-}(dv)$$
    This proves the general representation Equation \eqref{form of zeta}. 
\end{proof}

\begin{remark}
    When the equity price is homogeneous of degree one in the firm value,
     $S(t,v,\cdots)=v\cdot h(t,\cdots),$  
     the jump function satisfies 
     $J^S(v,\xi)=\tilde{S}_{T_n^-}^\ccH\cdot \Gamma,$ where $\Gamma(v,y,\xi):=\dfrac{\tilde G(v,y,\xi)}{v}$. So  Equation \eqref{form of zeta} becomes
     \begin{align*}
         \zeta^j_{T_n} = \dfrac{e^{-r(T-2T_n)}\displaystyle\int_{\mathbb R_{\ge 0}} \int_{\mathbb R^m} \bigg(u(T_n,v+\tilde G(v,Y_{T_n^-},\xi))-u(T_n,v)\bigg)\Gamma(v,Y_{T_n^-},\xi)F_\xi(d\xi)\pi_{T_n^-}(dv)}{\tilde{S}_{T_n^-}^\ccH\displaystyle\int_{\mathbb R_{\ge 0}} \int_{\mathbb R^m}\big(\Gamma(v,Y_{T_n^-},\xi)\big)^2F_\xi(d\xi)\pi_{T_n^-}(dv)}.
     \end{align*}
     In all numerical implementation, Equation \eqref{form of zeta} should be used, with $J^S$ evaluated from the particle approximation of $S$ at the post-jump value.
\end{remark}
The hedge ratio depends on the filtered distribution immediately before each scheduled announcement and therefore incorporates both the parameter uncertainty and disclosure risk. Unlike the full information delta hedge, the optimal strategy is computed with respect to investor beliefs rather than the latent firm value. Consequently, predictable disclosure events affect hedging not only through the jump dynamics of the asset process but also through the evolution of the posterior distribution.

\section{Numerical Implementation and Comparison}\label{sec:numerical}
This section illustrates the impact of predictable disclosure shocks on valuation, filtering, and hedging. We compare the proposed structural model with the partial information diffusion model of \cite{frey2009pricing}, which serves as the natural benchmark. The numerical experiments are designed to isolate the effect of predictable disclosure shocks on credit spreads, posterior beliefs, equity valuation, and hedging. Throughout, we specialize to the scalar setting $m=p=1$.

\subsection{Model Calibration}
\paragraph{Simulation framework} 
The firm's asset value is simulated under the risk-neutral measure $\mathbb Q$ using the dynamics introduced in Section \ref{sec:Model}. Between announcement dates, the diffusion component is discretized with an Euler-Maruyama scheme using a time step $\Delta t =0.01$ year, corresponding to approximately $2.5$ trading days. Scheduled disclosures occur quarterly at times $T_n =0.25n$ years, where the state-dependent jump mechanism introduced in Section \ref{sec:Model} is applied.

\paragraph{Benchmark model}
To isolate the effect of predictable disclosure shocks, we compare our model with the partial information structural model of \cite{frey2009pricing}. Their filter is implemented using the Markov chain approximation of proposition 3.4, whereas the proposed model employs the particle approximation of Section \ref{filtering sec}.

\paragraph{Calibration}
 Unless otherwise stated, all experiments use the baseline parameters reported in 
 Table \ref{Baseline Parameter}. The calibration represents a highly leveraged distressed firm and is chosen to generate realistic default probabilities and credit spreads while highlighting the economic effect of scheduled disclosure risk. The same parameter set is used throughout all numerical experiments unless explicitly varied.
 \begin{tiny}
\begin{table}[h!]
\centering
\caption{Baseline Parameter Values}
\label{Baseline Parameter}
\begin{tabular}{llcl}
\toprule
\textbf{Parameter} & \textbf{Symbol} & \textbf{Value} & \textbf{Description} \\
\midrule
Risk-free rate & $r$ & 3\% & Annual \\
Asset drift (risk-neutral) & $\mu_V$ & 5\% & Annual \\
Asset volatility & $\sigma_V$ & 25\% & Annual \\
Initial asset value & $V_0$ & 100 & \\
Initial liability & $K_0$ & 60 & Default barrier \\
Dividend smoothing & $\lambda$ & 0.3 & \\
Dividend noise & $\delta_k$ & Beta$(1,20)$ & Mean = 4.76\% \\
Observation noise & $\sigma_\eta$ & 0.10 & \\
Dividend financing & $\kappa$ & 1 & Pure asset sale \\
Loss given default & $\ell$ & 60\% & For CDS \\
\midrule
\multicolumn{4}{l}{\textit{Jump parameters (Extended model only):}} \\
Jump distribution & $\xi_n$ & Beta$(2,8)$ & Mean = 20\% write-down \\
Jump amplification & $\beta$ & 0.5 & Bad news multiplier \\
Observation threshold & $\bar{y}$ & 5 & For amplification \\
Observation persistence & $\alpha_y$ & 0.1 & In $f(V,Y) = \ln V + \alpha_y Y$ \\
\bottomrule
\end{tabular}
\end{table}
\end{tiny}
The same parameter set is used throughout all numerical experiments unless explicitly varied. 

\paragraph{Computational implementation}
All simulations are implemented in Python using NumPy and SciPy. Reported quantities are averaged over $10^3-10^4$ independent Monte Carlo paths. Simulations are performed on Apple M1 iMac with 16 GB RAM, and the complete numerical study requires less than one hour, including Hedging experiments.

\subsection{Numerical Methodology}
The numerical comparison requires solving two distinct filtering problems. The benchmark diffusion model of \cite{frey2009pricing} is implemented using the Markov-chain approximation described in their proposition 3.4, whereas the proposed model employs the particle filter developed in Section \ref{filtering sec}.
\begin{algorithm}[h!]
\caption{ Baseline filter, No Jumps (\cite{frey2009pricing} Filter)}
\label{algo 1}
\begin{algorithmic}
\State \textbf{Input:} Grid $\{m_1, \ldots, m_M\}$ on $[0.8K_0, 2V_0]$, observations $\{(d_k, \Delta Y_{T_k})\}_{k=1}^K$
\State \textbf{Initialize:} $q_j(0) = \mathbb{I}_{\{m_j \approx V_0\}} / \text{(normalization)}$
\For{$k = 1, 2, \ldots, K$}
    \State \textbf{Propagate between observations:} For $t \in (T_{k-1}, T_k)$, use trinomial transition matrix
    \State \quad $q(t) = P(t-T_{k-1})^{\lfloor (t-T_{k-1})/\Delta t \rfloor} q(T_{k-1})$
    \State \textbf{Update at observation $T_k$:}
    \State \quad Remove states below barrier: $q_j(T_k^-) = 0$ if $m_j \le K_{T_k}$
    \State \quad Dividend likelihood: $\ell_j^{(d)} = \nu_d(d_k | d_{k-1}, m_j)$
    \State \quad Observation likelihood: $\ell_j^{(Y)} = \phi(\Delta Y_{T_k} - \ln m_j; 0, \sigma_\eta^2)$ where $\phi$ is Gaussian density
    \State \quad Update: $q_j(T_k) \propto q_j(T_k^-) \cdot \ell_j^{(d)} \cdot \ell_j^{(Y)}$
    \State \quad Normalize: $q(T_k) \leftarrow q(T_k) / \sum_j q_j(T_k)$
\EndFor 
\State \textbf{Output:} Filtered mean $\hat{V}_t = \sum_{j=1}^M m_j q_j(t)$.%
\end{algorithmic}
\end{algorithm}
Algorithm \ref{algo 1} is used exclusively for the benchmark model and provides the reference prices and hedge ratios used throughout the numerical comparison.
\begin{tiny}
\scriptsize
\begin{algorithm}[h!]
\caption{Particle Filter for Extended Model}
\label{Particle filter}
\begin{algorithmic}
\scriptsize
\State \textbf{Input:} $N$ particles, observations $\{(d_k, \Delta Y_{T_k})\}_{k=1}^K$
\State \textbf{Initialize:}  ${X_0^{(i)} \sim \mathcal N(\ln V_0, 0.05^2)}$ for $i = 1, \ldots, N$, \quad $Y_0 = 0$, $w_0^{(i)} = 1/N$
\For{$k = 1, 2, \ldots, K$}
    \State \textbf{Propagate between observations} $t \in (T_{k-1}, T_k)$:
    \For{each time step $\Delta t$}
        \State \quad $X_t^{(i)} \leftarrow X_{t-\Delta t}^{(i)} + (\mu_V - \sigma_V^2/2)\Delta t + \sigma_V \sqrt{\Delta t} \cdot Z^{(i)}$, $Z^{(i)} \sim \mathcal N(0,1)$
        \State \quad $V_t^{(i)} = e^{X_t^{(i)}}$
        \State \quad \textbf{Check default:} If $V_t^{(i)} \le K_t$, set $w_t^{(i)} = 0$
    \EndFor
    \State
    \State \textbf{At observation time $T_k$:}
    \State \quad \textbf{Step 1: Pre-announcement update (Anticipation Effect)}
    \For{each particle $i$}
        \State \quad Store pre-jump value: $V_{T_{k^-}}^{(i)} = e^{X_{T_{k^-}}^{(i)}}$ 
        \State \quad Drawn $\xi_k^{(i)} \sim F_\xi$
        \State \quad Compute amplification: $A^{(i)} = 1 + \beta \cdot \mathbf{1}_{Y_{T_{k-1}} < \bar{y}}$
        \State \quad Drawn dividend $d_k^{(i)}\sim \nu_d\!\left(\cdot \mid d_{k-1},V_{T_k^-}^{(i)},Y_{T_{k-1}}\right)$.
        \State \quad Realized jump: ${G}^{(i)} = -\xi_k^{(i)} \cdot V_{T_{k^-}}^{(i)} \cdot A^{(i)}$
        \State \quad Update particle: $V_{T_k}^{(i)} = \max(V_{T_{k^-}}^{(i)} + {G}^{(i)} - \kappa {d}_k^{(i)},\varepsilon_\text{floor})$
        \State \quad Set $X_{T_k}^{(i)} = \ln V_{T_k}^{(i)}$
    \EndFor
    \State
    \State \quad \textbf{Step 2: Observation update (Innovation)}
    \State \quad Observe $\Delta Y_{T_k}$ and $d_k$
    \For{each particle $i$}
        \State \quad Expected observation: $\hat{Y}^{(i)} = \ln V_{T_{k}}^{(i)} + \alpha_y \cdot Y_{T_{k-1}}$
        \State \quad Observation likelihood: $\ell_Y^{(i)} = \phi(\Delta Y_{T_k} - \hat{Y}^{(i)}; 0, \sigma_\eta^2)$
        \State \quad Dividend likelihood: $\ell_d^{(i)} = \phi(d_k - d_k^{(i)}; 0, (0.1 d_k^{(i)})^2)$
        \State \quad Combined likelihood: $\ell^{(i)} = \ell_Y^{(i)} \cdot \ell_d^{(i)}$
        \State \quad Update weight: $w_{T_k}^{(i)} = w_{T_{k^-}}^{(i)} \cdot \ell^{(i)}$
    \EndFor
    \State \quad Normalize: $w_{T_k}^{(i)} \leftarrow w_{T_k}^{(i)} / \sum_j w_{T_k}^{(j)}$
    \State
    \State \quad \textbf{Step 3: Resample if needed}
    \State \quad Compute $N_{\text{eff}} = 1 / \sum_i (w_{T_k}^{(i)})^2$
    \If{$N_{\text{eff}} < N/3$}
        \State \quad Resample particles with replacement according to weights
        \State \quad Reset weights: $w_{T_k}^{(i)} = 1/N$
    \EndIf
    \State
    \State \quad Update observation history: $Y_{T_k} = Y_{T_{k-1}} + \Delta Y_{T_k}$
\EndFor
\State \textbf{Output:} Filtered distribution $\pi_t \approx \sum_{i=1}^N w_t^{(i)} \delta_{X_t^{(i)}}$.
\end{algorithmic}
\end{algorithm}
\end{tiny}
Algorithm \ref{Particle filter} implements the proposed nonlinear filtering framework with predictable disclosure updates and is used for all numerical experiments involving the extended model.
\newpage
\subsection{Effect of Predictable disclosure risk}
We first investigate how scheduled disclosures modify investors beliefs, credit spreads, and equity values relative to the diffusion-based benchmark of \cite{frey2009pricing}. The results demonstrate that predictable announcement risk introduces systematic effects that cannot arise in structural models with continuous information.

\paragraph{Posterior dynamics and sentiment-dependent jumps}
\begin{figure}[h]
\centering
\includegraphics[width=0.95\textwidth]{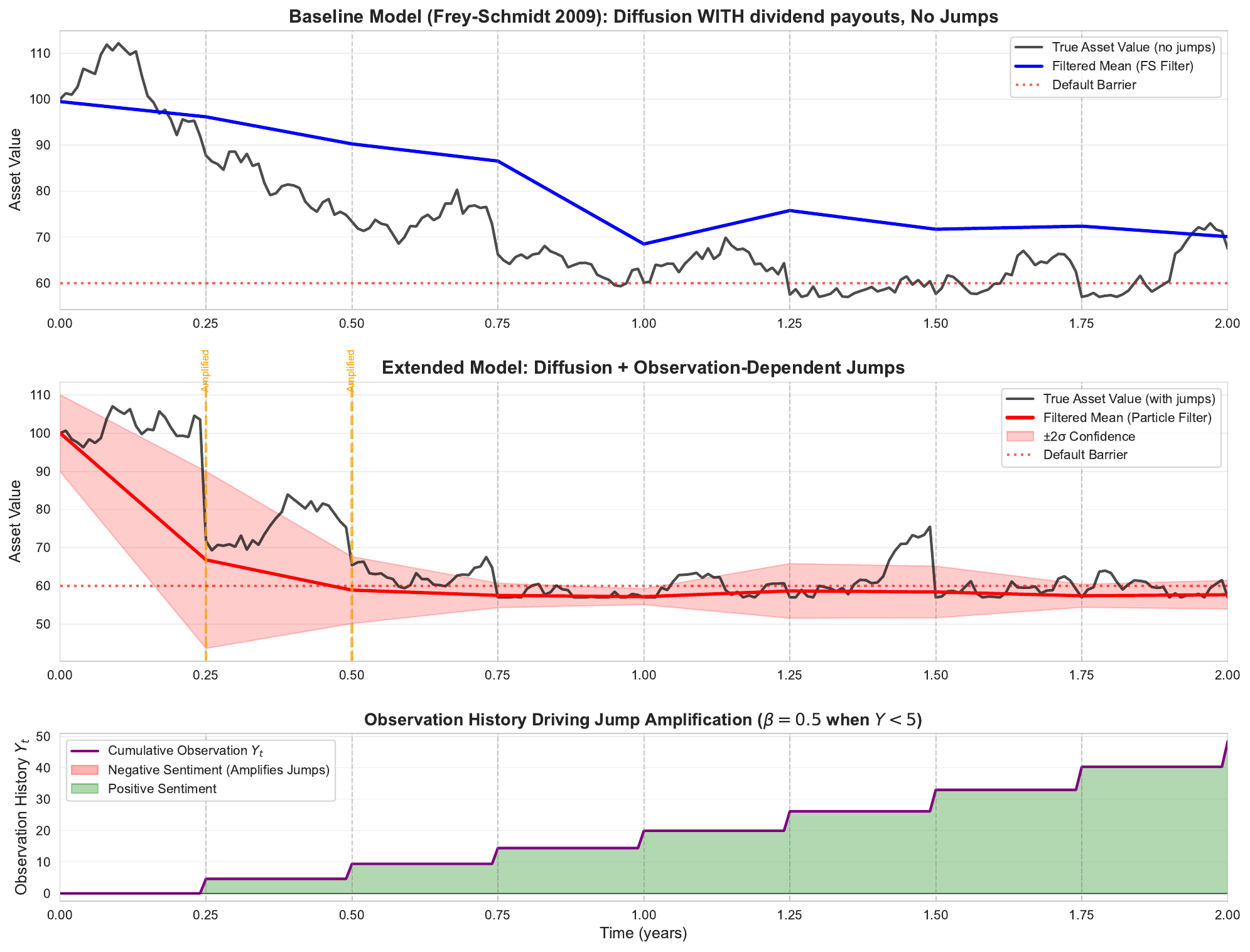}
\caption{Comparison of filtered asset values with negative sentiment scenario: Baseline vs Extended model with observation-dependent jumps.}
\label{fig:filtered_comparison}
\end{figure}
The comparison highlights three features. First, in the baseline model, the filtered mean tracks the true asset path closely throughout. Without jumps, dividend observations provide only weak continuous updating and the filter has little uncertainty to resolve between announcement: the absence of informational shocks leaves the posterior largely undisturbed between scheduled dates. Second, the extended model (middle panel) presents a richer picture. Between announcements the filter faces genuine uncertainty about the firm's value, reflected in the wide $\pm 2\sigma$ confidence bands; at each disclosure date, however, that uncertainty collapses sharply as the observation increment $\Delta Y_{T_n}$ is revealed. The first two announcements, at $t=0.25$ and $t=0.5$, trigger amplified downward revisions by factor $1+\beta =1.5$ because cumulative observation history $Y_t$ lies the threshold $\bar y=5$ at those dates, activating the $\beta$-amplification channel visible in the bottom panel. later jumps are smaller as sentiment recovers and posterior updates become less dramatic. Third, the bottom panel quantifies the sentiment dynamics.

\emph{Economic interpretation:} Predictable jumps transform scheduled announcement dates from mechanical dividend payouts into genuine informational events. In the baseline, corporate dates reduce asset value by a known amount and leave the filter largely undisturbed. In the extended model, the same dates concentrate both default risk and filtering uncertainty simultaneously, producing the sharp posterior revisions and pre-announcement spread widening that characterize observed credit market behaviour around earnings releases and rating reviews \cite{kothari2009effect}. The asymmetry between early and late announcements is not a calibration artifact but a structural feature of the model: adverse news compounds when the informational environment is already pessimistic. To be more precise, the figure reveals that the filtering problem in this framework carries two coupled unknowns rather than one: the investor must simultaneously learn the firm's unobservable asset value and the prevailing sentiment regime driving jump amplification, and it is precisely this coupling that standard diffusion-based filtering models cannot capture.

\paragraph{Credit spread dynamics}
Credit spreads are computed for maturities between $0.25$ and $5$ years under both neutral and negative sentiment using Monte Carlo simulation with $50,000$ paths.
\begin{figure}[h!]
\centering
\includegraphics[width=0.95\textwidth]{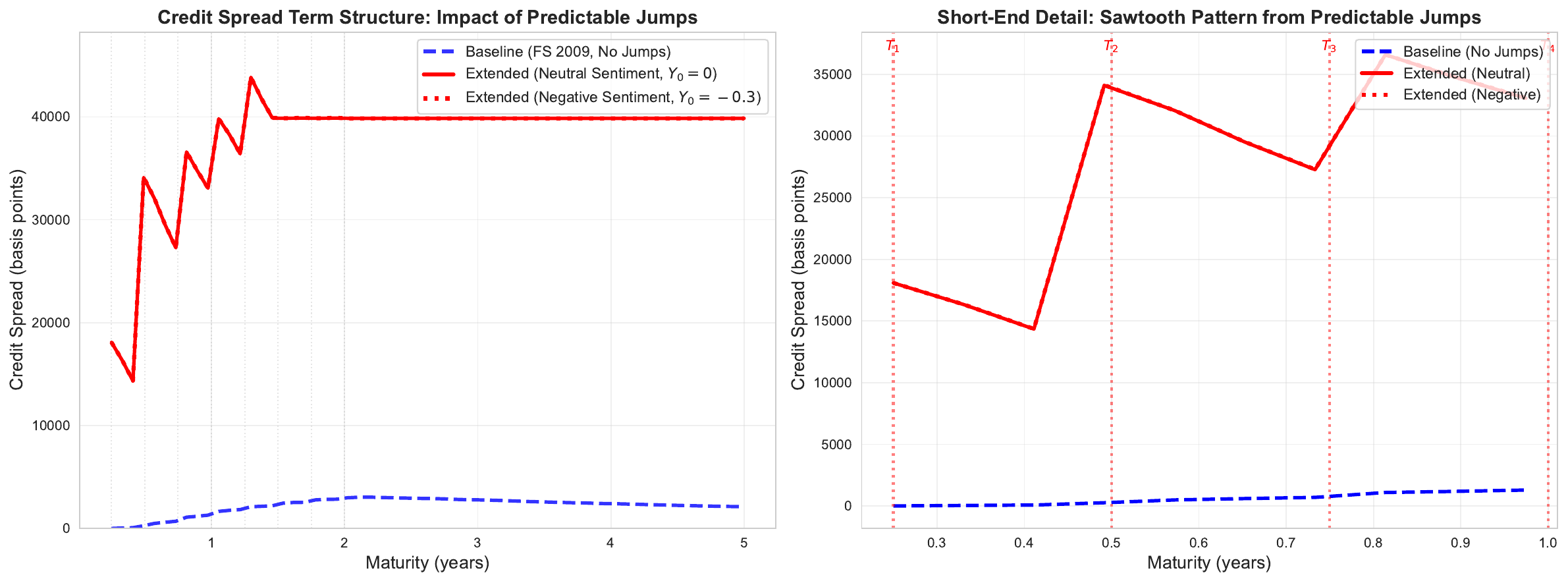}
\caption{Credit spread term structures with observation-dependent jumps. 
}
\label{fig:spread_comparison_final}
\end{figure}
Figure \ref{fig:spread_comparison_final} reveals two effects that are absent from the diffusion only baseline. Predictable disclosures generate a characteristic sawtooth term structure: spreads widen prior to scheduled announcements as investors anticipate adverse disclosure shocks and fall immediately afterwards once uncertainty is resolved. Negative sentiment amplifies this mechanism through larger expected jump sizes, leading to systematically higher spreads during the early announcement periods. Neither feature is present in the diffusion benchmark, where information arrives continuously. 

\paragraph{Equity Valuation with Path-Dependent Jumps}
\begin{figure}[h!]
\centering
\includegraphics[width=0.75\textwidth]{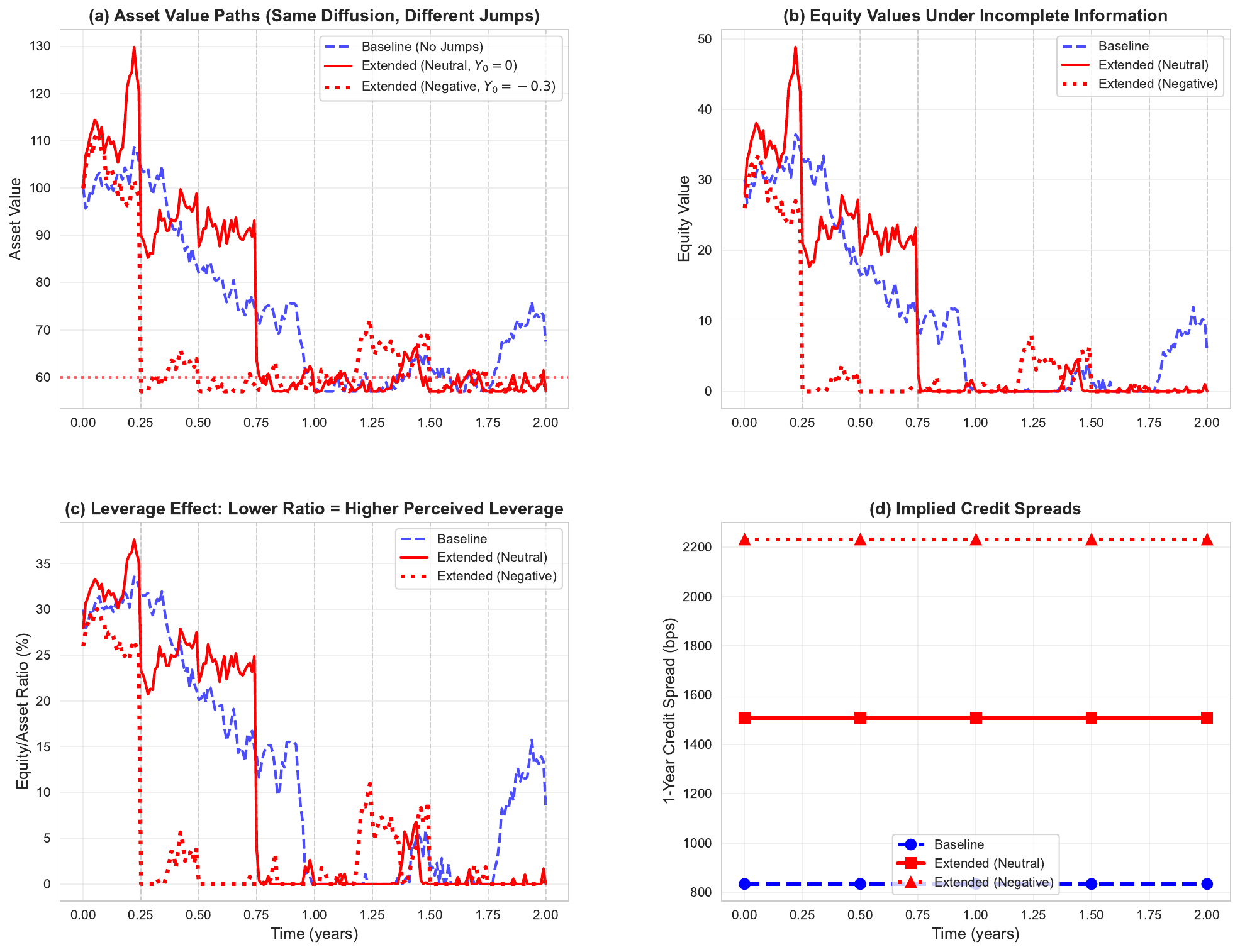}
\caption{Equity valuation with observation-dependent jump risk.}
\label{fig:equity_paths}
\end{figure}
Figure \ref{fig:equity_paths} illustrates the effect of predictable disclosure risk introduces on equity valuation. Since all scenarios share the same underlying diffusion realization, differences in equity values arise solely from anticipated disclosure shocks. Negative sentiment generates a persistent valuation discount because investors continue to assign positive probability to future amplified disclosures even after the asset paths have largely reconverged. This effect increases perceived leverage and translates into wider credit spreads relative to the diffusion benchmark. Consequently, models that ignore predictable disclosure risk systematically overvalue equity and underestimate credit risk during periods of adverse market sentiment.

\subsection{Filtering Performance}\label{filtering_performance}
 We next evaluate the accuracy of the proposed nonlinear filter and examine the relationship between posterior uncertainty and default risk.
\paragraph{Filtering estimation errors}
\begin{figure}[h!]
\centering
\includegraphics[width=0.75\textwidth]{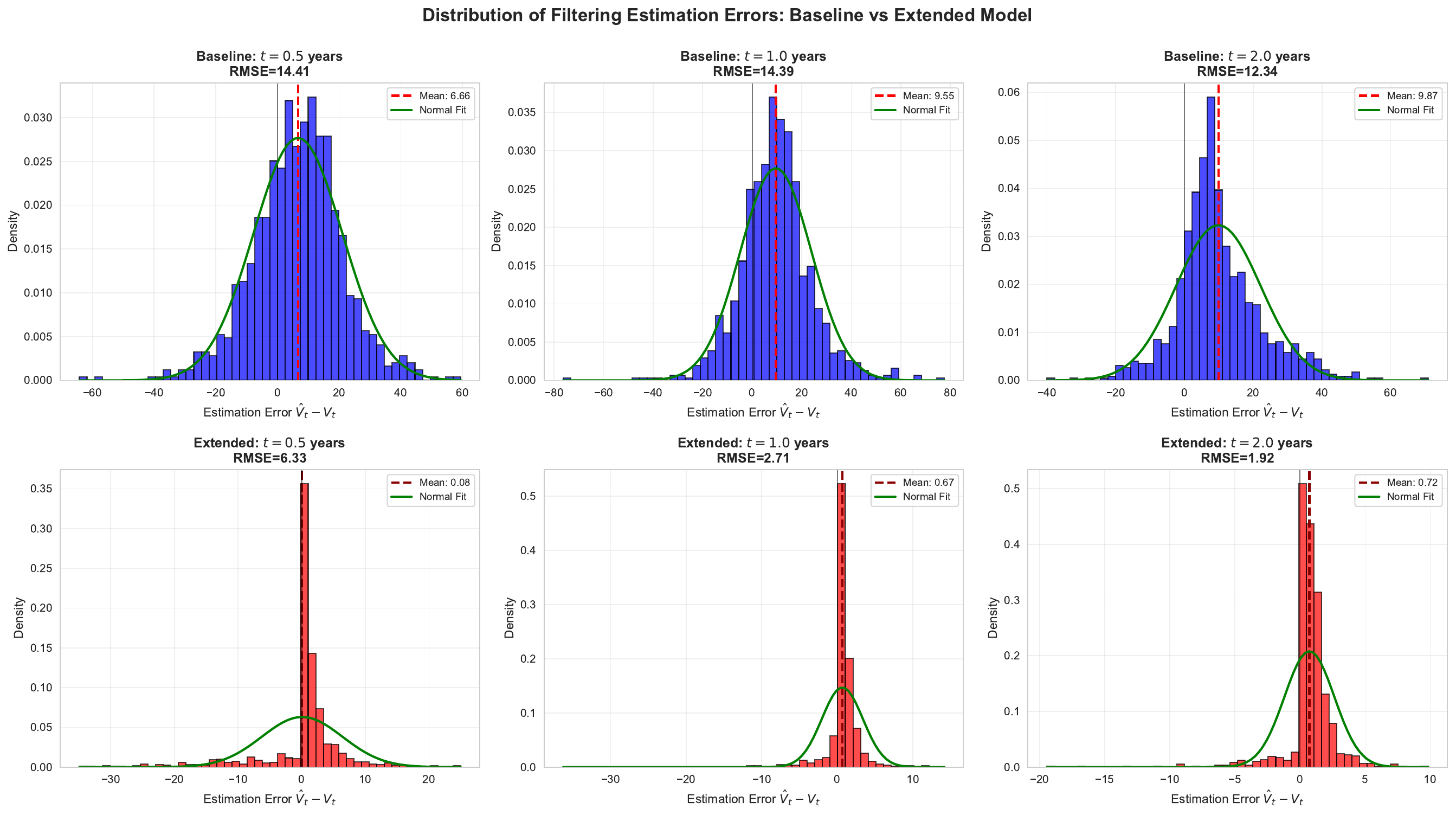}
\caption{Distribution of filtering estimation errors $\hat{V}_t - V_t$ at different time horizons (1000 simulated paths).}
\label{fig:estimation_errors}
\end{figure}

Figure \ref{fig:estimation_errors} compares the distribution of filtering errors at three horizons. The diffusion benchmark exhibits a gradual increase in estimate error, with the RMSE rising from $18.53$ after six months to $25.70$ at two years as uncertainty accumulates between observations. In contrast, the proposed model achieves progressively lower errors, with the RMSE decreasing from $6.33$ to $1.92$ over the same period. Although scheduled disclosures introduce temporary jump uncertainty, they also provide highly informative observations that substantially reduce posterior uncertainty. Repeated Bayesian updates therefore offset the accumulation of diffusion noise and improve long-run filtering accuracy.
\newpage
\paragraph{Equity valuation under partial information}
Figure \ref{fig:equity_scatter} illustrates how predictable disclosure risk alters the relationship between firm value and equity value. Under the diffusion benchmark, equity is closely related to the current asset value, producing an almost linear relationship. In contrast, the proposed model exhibits both a flatter slope and greater dispersion because equity prices incorporate the possibility of future sentiment-dependent disclosure shocks. Consequently, firms with identical asset values may have substantially different equity valuations depending on their posterior information state, generating an additional discount that is absent from the benchmark model.

\begin{figure}[h]
\centering
\includegraphics[width=0.85\textwidth]{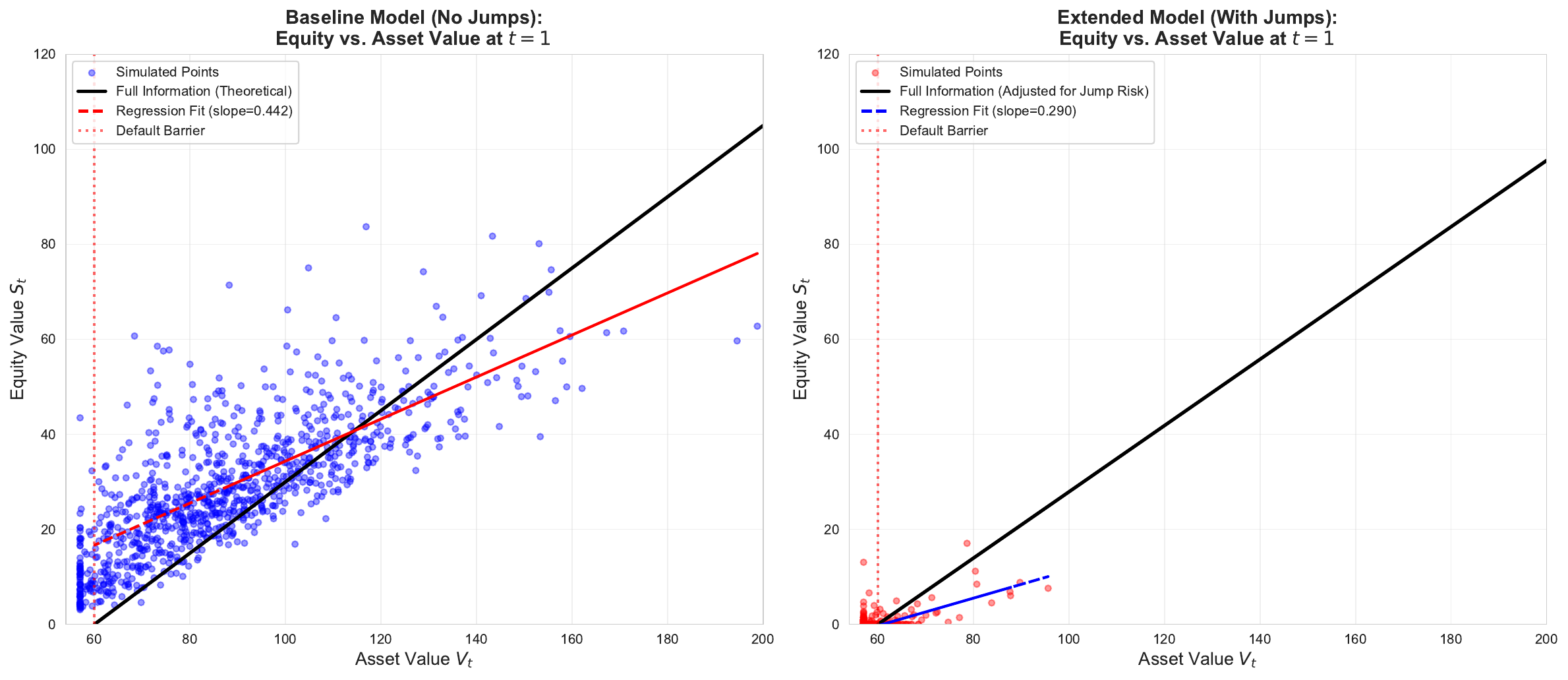}
\caption{Equity vs. asset value scatter plots at $t=1$ year ($1000$ simulated paths). }
\label{fig:equity_scatter}
\end{figure}
These results imply that disclosure risk weakens the sensitivity of equity to improvements in firm value and therefore alters incentives for financing and investment decisions. In particular, firms have greater incentives to issue equity immediately after favorable disclosures, when announcement uncertainty has been resolved.

\subsection{Hedging performance}\label{sec:hedging numeric}
Finally, we illustrate the practical implications of the local risk-minimizing strategy in Section \ref{sec 5: Hedging}. 
We compared the proposed incomplete-information strategy with the corresponding full-information and diffusion based benchmark hedges to quantify the impact of predictable disclosure risk on hedging performance.

The numerical experiment considers the contingent claim $H = (V_T - 1.5K_0)^+$ hedged dynamically using the firm's equity. At each scheduled disclosure date, the local risk-minimizing hedge ratio is computed from Proposition \ref{hedge ratio} using the pre-announcement filter $\pi_{T_n^-}$. 

\begin{figure}[h!]
    \centering
    \includegraphics[width=0.85\textwidth]{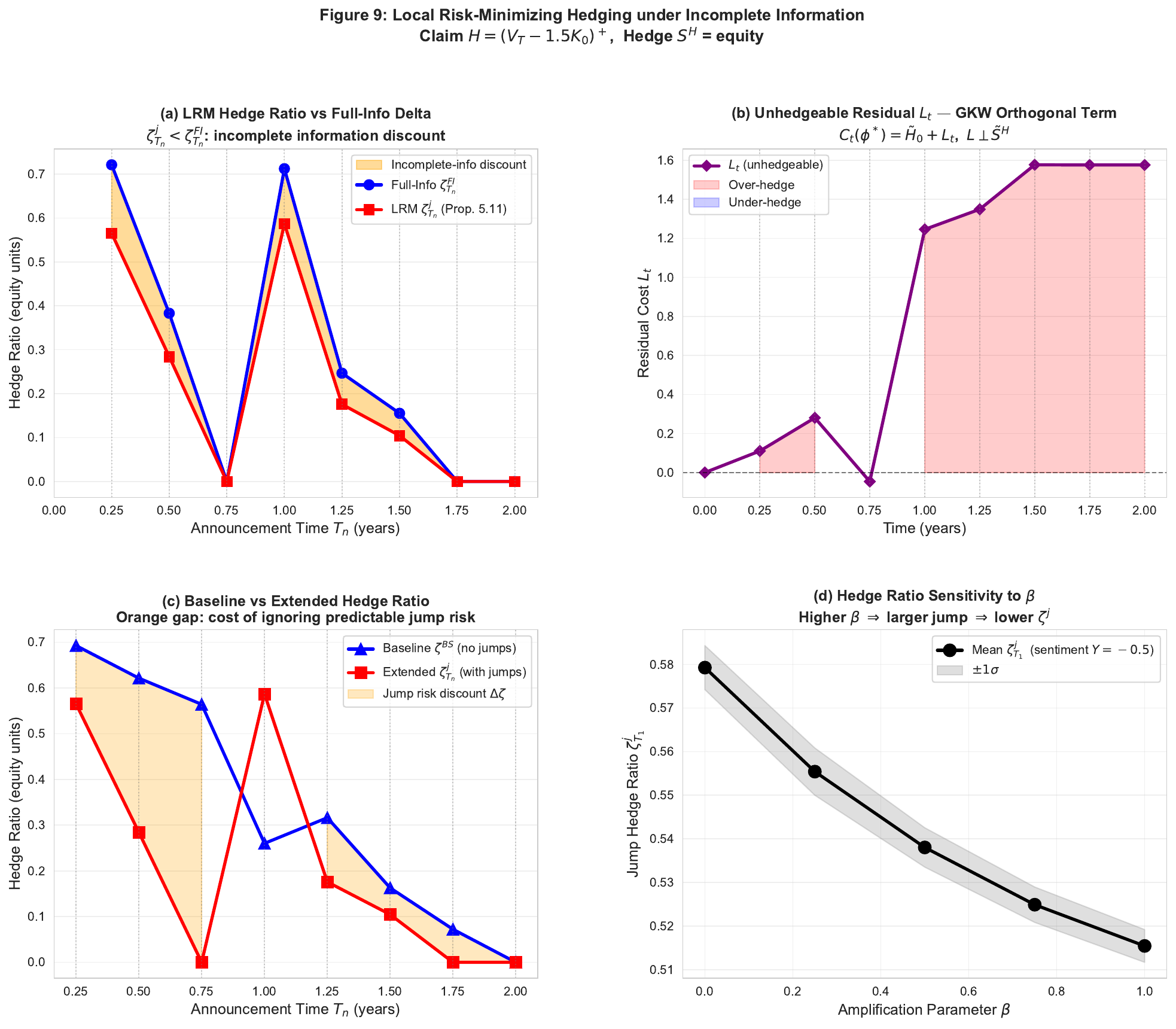}
    \caption{Hedging under incomplete information with predictable jumps. 
    The claim is $H = (V_T - 1.5K_0)^+$ and the hedge instrument is equity $S^\ccH$.
    }
    \label{fig:hedging_lrm}
\end{figure}

Figure \ref{fig:hedging_lrm} illustrates four consequences of predictable disclosure risk. First, the incomplete information hedge ratio is systematically lower than its full-information counterpart because investors hedge against the posterior distribution rather than the latent asset value. Second, the orthogonal residual increases after major disclosure dates, reflecting the portion of jump risk that cannot be replicated using the traded equity alone. Third, the diffusion benchmark either over-or under-hedges depending on the contract's moneyness, demonstrating that ignoring predictable jumps introduces systematic hedging errors. Finally, stronger sentiment amplification reduces the optimal hedge ratio by increasing the sensitivity of equity prices to disclosure shocks.

These results confirm that predictable disclosure risk affects hedging through both the jump dynamics of firm value and the Bayesian evolution of investor beliefs. Even when disclosure dates are known in advance, incomplete information generates residual risk that cannot be eliminated through dynamic trading. 

\subsection{Sensitivity analysis}
We next investigate the robustness of the model by varying the amplification parameter, the observation noise, and the disclosure frequency while holding all remaining parameters at their baseline values. 
We present in Table \ref{tab:parameter sensitivity} the comprehensive sensitivity analysis.
\begin{table}[http]
\centering
\caption{\textbf{Parameter sensitivity analysis: Impact on key model outputs}}
\label{tab:parameter sensitivity}
\resizebox{\linewidth}{!}{
\begin{threeparttable}
\begin{tabular}{llccccc}
\toprule
\multicolumn{2}{l}{Parameter \& Value}
  & \begin{tabular}[c]{@{}c@{}}1Y Spread\\(bps)\end{tabular}
  & \begin{tabular}[c]{@{}c@{}}1Y Default\\Prob (\%)\end{tabular}
  & \begin{tabular}[c]{@{}c@{}}Filter\\RMSE\end{tabular}
  & \begin{tabular}[c]{@{}c@{}}Equity\\Value\end{tabular}
  & \begin{tabular}[c]{@{}c@{}}Jump Risk\\Share (\%)\end{tabular}\\
\midrule
\multicolumn{7}{l}{\textit{A.~Amplification Factor $\beta$ (Observation Dependence)}}\\
\midrule
          & $\beta = 0.00$ & 31{,}512 & 95.7 & 6.2 & 6.36 & 77.1 \\
          & $\beta = 0.25$ & 36{,}100 & 97.3 & 5.7 & 5.05 & 81.6 \\
Baseline  & $\beta = 0.50$ & 39{,}900 & 98.2 & 4.9 & 4.05 & 85.3 \\
          & $\beta = 0.75$ & 43{,}900 & 98.8 & 4.3 & 3.12 & 87.9 \\
          & $\beta = 1.00$ & 46{,}300 & 99.0 & 3.8 & 2.59 & 89.8 \\
\midrule
\multicolumn{2}{l}{Change (\%) $0\to1$}
  & $+47\%$ & $+3\%$ & $-39\%$ & $-59\%$ & $+17\%$\\
\midrule
\multicolumn{7}{l}{\textit{B.~Observation Noise $\sigma_\eta$ (Information Quality)}}\\
\midrule
          & $\sigma_\eta = 0.05$ & 39{,}318 $\pm$ 1,255 & 97.78 $\pm$ 0.20 &  2.73 & 4.167 $\pm$ 0.071 & 84.9 $\pm$ 0.6 \\
          & $\sigma_\eta = 0.10$ & 39{,}318 $\pm$ 1,255 & 97.78 $\pm$ 0.20 &  2.71 & 4.167 $\pm$ 0.071 & 84.9 $\pm$ 0.6 \\
Baseline  & $\sigma_\eta = 0.15$ & 39{,}127 $\pm$ 1,249 & 97.70 $\pm$ 0.21 &  3.04 & 4.187 $\pm$ 0.071 & 84.8 $\pm$ 0.5 \\
          & $\sigma_\eta = 0.20$ & 38{,}703 $\pm$ 1,064 & 97.70 $\pm$ 0.20 &  3.47 & 4.206 $\pm$ 0.071 & 84.7 $\pm$ 0.5 \\
          & $\sigma_\eta = 0.25$ & 38{,}644 $\pm$ 1,078 & 97.68 $\pm$ 0.21 & 3.76 & 4.224 $\pm$ 0.072 & 84.5 $\pm$ 0.5 \\
\midrule
\multicolumn{2}{l}{Change (\%) $0.05\to0.25$}
  & $-1.7\%$ & $-0.1\%$ & $+38\%$ & $+1.4\%$ & $-0.5\%$\\
\midrule
\multicolumn{7}{l}{\textit{C.~Jump Frequency (Observation Frequency)}}\\
\midrule
Baseline & Quarterly (4/yr)  & 37{,}297 $\pm$ 2,252         & 97.60$\pm$ 0.54           & 2.75 & 4.162 $\pm$ 0.167 & 85.5 $\pm$ 1.2 \\
         & Monthly  (12/yr)  & $>$200{,}000     & 100.0 $\pm$ 0.00           & 0.79 & 3.753 $\pm$ 0.114 & 94.5 $\pm$ 0.6 \\
         & Weekly   (52/yr)  & $>$200{,}000     & $100$ $\pm$ 0.00 & 0.17 & 3.460 $\pm$ 0.123 & 100.0 $\pm$ 0.0 \\
\midrule
\multicolumn{2}{l}{Change (\%) $4\to52$/yr}
  & $\uparrow$ & $\uparrow$ & $-94\%$ & $-17\%$ & $+17\%$\\
\midrule
\multicolumn{7}{l}{\textit{D.~Combined Effects: Worst vs Best Case}}\\
\midrule
Best Case
  & $\beta{=}0,\;\sigma_\eta{=}0.05$, 4/yr
  & 29,004 $\pm$ 1,282 & 94.57 $\pm$ 0.71 &  3.23 & 6.46 $\pm$ 0.23 & 76.8 $\pm$ 1.4 \\
Baseline
  & $\beta{=}0.5,\;\sigma_\eta{=}0.10$, 4/yr
  & 44{,}228 $\pm$ 2,472 & 98.8 $\pm$ 0.30 &  2.41 & 3.94 $\pm$ 0.14 & 84.3 $\pm$ 1.1 \\
Worst Case
  & $\beta{=}1.0,\;\sigma_\eta{=}0.25$, 52/yr
  & $>$200{,}000 & 100 $\pm$ 0.00 & 0.27 & 2.437 $\pm$ 0.088 & 100 $\pm$ 0.0 \\
\midrule
\multicolumn{2}{l}{Range Best $\to$ Worst}
  & $\uparrow\uparrow$ & $+6\%$ & $-92\%$ & $-62\%$ & $+30\%$\\
\bottomrule
\end{tabular}
\begin{tablenotes}
\footnotesize
\item[1] \textbf{Jump Risk Share} = percentage of defaults attributable to
  jump-induced barrier crossings at scheduled announcement dates.
\item[2] \textbf{Filter RMSE} = root-mean-squared error of the particle filter
  estimate of $V_t$ at $t=1$ year over all simulated paths. Filter RMSE values in Panels B-D use full (non-subsampled) particle filtering on every simulated path ($N=2000$ particles per path). Panel B uses common random numbers across $\sigma_\eta$ values to enable paired-difference significance testing (see note 3); Panel C and D use independent batch means, since scenarios there differ in disclosure frequency and cannot be path-paired.
\item[3] \textbf{Inference for Panel B}
  $\sigma_\eta$ has small but statistically significant effect on filtered outcomes. Spread, default probability, and equity move monotonically but modestly over the tested range (a $1$-$2\%$ change from $\sigma_\eta=0.05$ to $0.25$), consistent with the mechanism whereby larger observation noise occasionally allows the sentiment indicator $Y$ to cross the reassurance threshold $\bar{y}$ sooner than its deterministic drift alone would imply, slightly reducing the incidence of amplified disclosure shocks. Filtering accuracy is essentially unaffected between $\sigma_\eta=0.05$ and $0.10$ (paired $t=-0.29$, not significant) but degrades significantly and monotonically for $\sigma_\eta \ge 0.10$ (RMSE rising from $2.71$ to $3.76$, a $\sim 38\%$ increase; paired-difference $t$-statistics of $7.4$-$12.6$ across successive steps, $t=11.75$ comparing the extremes). All reported values use common random numbers across $\sigma_\eta$ with $N=5,000$ simulated paths; standard errors are computed via batch means ($20$ batches).
\end{tablenotes}
\end{threeparttable}}
\end{table}
The amplification parameter $\beta$ has the largest impact on credit risk. Increasing $\beta$ raises credit spreads and default probabilities while reducing equity values, reflecting the stronger effect of adverse disclosures on firm value. At the same time, filtering accuracy improves because larger disclosure shocks provide more informative observations.

Observation noise primarily affects filtering accuracy. As expected, larger observation noise substantially increases the RMSE of the particle filter while leaving equilibrium prices almost unchanged, since the underlying asset dynamics remain unaffected.
Increasing the disclosure frequency strengthens the role of predictable jumps. More frequent announcements increase the proportion of jump-induced defaults and reduce equity values, confirming that scheduled disclosures become the dominant driver of credit risk under frequent information releases.

\begin{figure}[h!]
    \centering
    \includegraphics[width=0.95\linewidth]{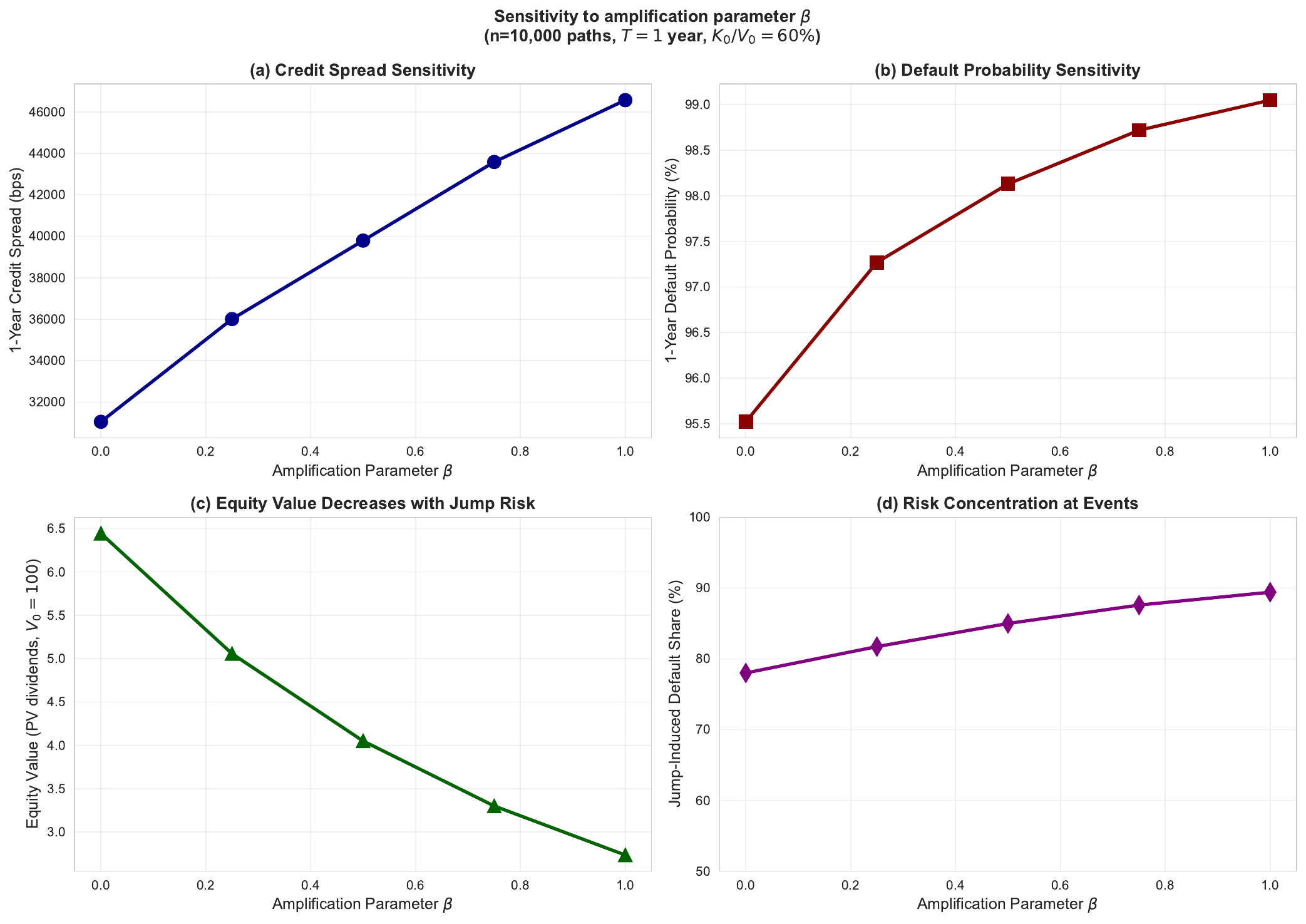}
    \caption{Sensitivity to the amplification parameter $\beta$}
    \label{fig:beta_sensitivity}
\end{figure}
 Figure \ref{fig:beta_sensitivity} further illustrates the nonlinear effect of the amplification parameter. Credit spreads increase convexly with $\beta$, while default probabilities approach saturation because of the highly leveraged calibration. In contrast, equity values decrease almost linearly, reflecting the increasing expected cost of disclosure induced losses. The proportion of jump induced defaults also rises steadily, confirming that predictable disclosure events dominate continuous default as amplification increase.

 With all above results, the numerical experiments consistently demonstrate that predictable disclosure risk influences pricing, filtering and hedging simultaneously. Scheduled announcements not only improve state estimation through Bayesian updating but also generate economically significant effects on credit spreads, equity values, and hedging strategies that are absent from diffusion based structural models.

\appendix
\section{Mathematical Appendix}\label{Appendix A}
\begin{proof}[Proof of Proposition \ref{filter convergence}]
The proof proceeds in four steps. We first construct the unnormalized filtering distribution and its particle approximation. We then establish boundedness of the likelihood weights. Third, we invoke the standard convergence theory for sequential Monte Carlo methods. Finally, we pass to the normalized filter.

First, we fix $t>0$ and let $n_t=\#\{T_m\le t\}$ denote the number of predictable announcement dates before time $t$. Since $n_t<\infty$ almost surely by Assumption \ref{filter condition}.3, only finitely many Bayesian updates occur on every finite time interval. Let $\tilde V$ denote the reference signal process governed by Equation \eqref{eq:asset_dynamics}. The associated unnormalized filtering distribution is

$$\tilde\pi_t(\varphi) = E^{\mathbb Q}[\varphi(\tilde V_t)\mathcal L_t|\, \ccH_t],$$
where the likelihood process is
$$\mathcal L_t=\prod_{T_m\le t}\ell_m\mathbf1_{\{\tilde V_s>K_s,\;0\le s\le t\}},$$
with
$\ell_m=\nu_d(d_m|d_{m-1},\tilde V_{T_m^-},Y_{T_m^-})p_\eta\left(\Delta Y_{T_m}-f(\tilde V_{T_m^-},Y_{T_m^-})\right).$
The corresponding normalized filter is
$$\pi_t(\varphi)=\frac{\tilde\pi_t(\varphi)}{\tilde\pi_t(1)},$$
where $\tilde\pi_t(1)>0$ almost surely on the event of survival.
The particle approximation is
$\tilde\pi_t^N(\varphi)=\dfrac1N\sum_{i=1}^Nw_t^{(i)}\varphi(V_t^{(i)}),$
with normalized estimator
$$\pi_t^N(\varphi)=\dfrac{\tilde\pi_t^N(\varphi)}{\tilde\pi_t^N(1)}.$$

Second, by Assumption \ref{filter condition}.4, $p_\eta(\eta)\le C_\eta,$ while Assumption \ref{filter condition}.5 yields $$\nu_d(d_m|d_{m-1},v,y)\le C_d$$ uniformly over all admissible $(v,y)$. Hence every likelihood factor satisfies
$\ell_m\le C_dC_\eta=:C_\ell.$
Since the survival indicator is bounded by one,
$$0\le w_t^{(i)}=\prod_{T_m\le t}\ell_m^{(i)}\mathbf1_{\{\tilde V_s^{(i)}>K_s\}}\le C_\ell^{\,n_t}=:M_t,$$
where $M_t<\infty$ almost surely because $n_t<\infty$. Furthermore, since $\varphi\in C_b^2(\mathbb R)$, $|\varphi(v)|\le \|\varphi\|_\infty,$ and therefore
$\left|w_t^{(i)}\varphi(V_t^{(i)})\right|\le M_t\|\varphi\|_\infty.$
Consequently, the particle potentials are uniformly bounded on every finite time interval.

Third, the particle system generated by Algorithm \ref{Particle filter} is therefore a standard Feynman-Kac particle approximation of the nonlinear filtering distribution with bounded potential functions.
Under Assumption \ref{filter condition}, the signal process possesses finite moments of order $p>2$, the transition kernel is Lipschitz, the observation likelihood is bounded, and only finitely many predictable jump updates occur on every compact time interval. These conditions are precisely those required in the standard convergence theory for sequential Monte Carlo methods (see, e.g., Del Moral \cite{del2004feynman}, Crisan and Doucet \cite{crisan2002survey}, and Bain and Crisan \cite{bain2009crisan}). Hence, for every bounded test function $\varphi$,
 $\tilde\pi_t^N(\varphi)\longrightarrow \tilde\pi_t(\varphi), \text{almost surely}$ as $N\rightarrow\infty$.
Applying the same result with $\varphi\equiv1$ gives
$$\tilde\pi_t^N(1)\longrightarrow\tilde\pi_t(1)>0\qquad \text{almost surely}.$$

Fourth, since both numerator and denominator converge almost surely and $\tilde\pi_t(1)>0,$ the Continuous Mapping Theorem yields
\begin{align*}
    \pi_t^N(\varphi)=\dfrac{\tilde\pi_t^N(\varphi)}{\tilde\pi_t^N(1)}\longrightarrow \dfrac{\tilde\pi_t(\varphi)}{\tilde\pi_t(1)}= \pi_t(\varphi)
\end{align*}
almost surely as $N\rightarrow\infty$.
\end{proof}

\section*{Acknowledgments}
The author thanks Thorsten Schmidt for his continuous strong support and for reading the earlier version. Also  Wilfried Kuissi Kamdem for helpful discussions and for reading the earlier version. FBTN is funded the Deutsche Forschungsgemeinschaft (DFG, German Research Foundation) – Project-ID 499552394 – SFB 1597 Small Data.

\bibliographystyle{plain}
\bibliography{references}

\end{document}